\documentclass[a4paper,fleqn,11pt]{article}

\usepackage{amsmath}
\usepackage{amsthm}
\usepackage{amssymb}
\usepackage{mathtools}
\usepackage{accents}
\usepackage{float}
\usepackage{comment}

\usepackage[a4paper,top=3cm, bottom=3cm, left=3cm, right=3cm]{geometry}
\usepackage[shortlabels,inline]{enumitem}

\usepackage[pdftex]{graphicx}

\usepackage[pdftex,ocgcolorlinks,pagebackref=false]{hyperref}
\usepackage{authblk}
\usepackage{tikz}

\setlist[enumerate,1]{label={(\roman*)}}

\usepackage[capitalise,noabbrev]{cleveref}

\Crefname{property}{Property}{Properties}

\theoremstyle{plain}
\newtheorem{theorem}{Theorem}[section]
\newtheorem{lemma}[theorem]{Lemma}
\newtheorem{proposition}[theorem]{Proposition}
\newtheorem{corollary}[theorem]{Corollary}
\newtheorem{remark}[theorem]{Remark}

\theoremstyle{definition}
\newtheorem{definition}[theorem]{Definition}

\newcommand{\norm}[2][]{\left\|#2\right\|_{#1}}

\newcommand{\ket}[1]{\left|#1\right\rangle}
\newcommand{\bra}[1]{\left\langle #1\right|}
\newcommand{\ketbra}[2]{\left|#1\middle\rangle\!\middle\langle#2\right|}
\newcommand{\braket}[2]{\left\langle#1\middle|#2\right\rangle}
\newcommand{\vectorstate}[1]{\ketbra{#1}{#1}}

\newcommand{\Diag}{\textnormal{Diag}}

\newcommand{\TV}{\textnormal{TV}}
\newcommand{\unitvectors}{\mathbb{S}}
\newcommand{\T}{\textnormal{T}}

\newcommand{\GHZ}{\textnormal{GHZ}}
\newcommand{\rank}{\textnormal{rank}}

\newcommand{\Hilbert}{\mathcal{H}}

\DeclareMathOperator{\boundeds}{\mathcal{B}}
\DeclareMathOperator{\Tr}{Tr}

\newcommand{\entropy}{H}

\newcommand{\relativeentropy}[3][]{\mathop{D_{#1}}\mathopen{}\left(#2\middle\|#3\right)\mathclose{}}
\newcommand{\binaryrelativeentropy}[3][]{\mathop{d_{#1}}\mathopen{}\left(#2\middle\|#3\right)\mathclose{}}

\newcommand{\distributions}[1][]{\mathcal{P}_{#1}}
\DeclareMathOperator{\states}{\mathcal{S}}
\newcommand{\unittensor}[1]{\langle{#1}\rangle}

\newcommand{\reals}{\mathbb{R}}
\newcommand{\complexes}{\mathbb{C}}
\newcommand{\naturals}{\mathbb{N}}
\newcommand{\integers}{\mathbb{Z}}
\newcommand{\nonnegativereals}{\mathbb{R}_{\ge 0}}

\newcommand{\LOCC}{\xrightarrow{\text{LOCC}}}

\DeclareMathOperator{\supp}{supp}

\title{Asymptotic equipartition property of subadditive multipartite entanglement measures on pure states}

\author[1]{D\'avid Bug\'ar}
\affil[1]{Department of Algebra and Geometry, Institute of Mathematics, Budapest University of Technology and Economics, M\H uegyetem~rkp. 3., H-1111 Budapest, Hungary.}
\date{}

\begin{document}

\maketitle

\begin{abstract}
We investigate the asymptotic equipartition property (AEP) in the context of multipartite entanglement measures on pure states. Specifically, we formulate AEP for subadditive entanglement measures that admit certain weak conditions. This is motivated by the uniqueness of the entanglement entropy in the asymptotic limit in the bipartite case. On the other hand, its operational relevance comes from the $\text{LOCC}_q$ scenario (asymptotic local operations and classical communication with a sublinear amount of quantum communication). Analogously to the classical AEP, we prove that the regularization of smooth weakly additive entanglement measures (subject to some weak extra conditions) yields weakly additive and asymptotically continuous entanglement measures. Then evaluate the mentioned regularization and smoothing on known R\'enyi type
multipartite entanglement measures, showing that the resulting regularized entanglement measures reduce to convex combinations of bipartite entanglement entropies.
\end{abstract}

\section{Introduction}

The asymptotic equipartition property (AEP) describes how long sequences of random variables are distributed \cite{cover2012elements}. Let $X_1, X_1, \dots, X_n$ be i.i.d. random variables distributed according to the probability distribution $P$ over the set $\mathcal{X}$. For large $n$ the sequences $X=(X_1,\dots, X_n)$ are concentrated on a typical set consisting of approximately $2^{n\entropy(P)}$ elements with probability close to $2^{-n\entropy(P)}$, where
$H(P)$ is the Shannon entropy of $P$.
This property comes handy in applications such as source compression. This tells us how many bits we need to store $n$ random bits sampled from the same distribution, tolerating a small error.  

An alternative formulation of this property is given by the R\'enyi entropies (to be introduced later) \cite{cover2012elements,tomamichel2009fully}. Informally, the smooth R\'enyi entropies $\entropy_\alpha^\epsilon(P)$ are defined as the infima (for $\alpha< 1$) and suprema ($\alpha>1$) of the respective entropies by ignoring $\epsilon$ probability. The AEP in terms of 
R\'enyi entropies is then formulated by the equality
\begin{equation}\label{eq:AEPentropic}
\lim_{\varepsilon \to 0} \lim_{n \to \infty} \frac{1}{n} H_\alpha^{\varepsilon}(P^{\otimes n}) = H(X),
\end{equation}
where $P^{\otimes n}$ denotes the product distribution $P\otimes\dots\otimes P$ over $\mathcal{X}^n$. 
Certain generalizations of the classical AEP were made in the framework of quantum information theory.
In \cite{tomamichel2009fully} they prove a generalization of the asymptotic equipartition property, when both the outcome and the side information of an of the experiment is quantum, while in \cite{fang2024generalized} a generalized quantum AEP was shown beyond the i.i.d. framework, showing that under some assumptions, all operationally relevant divergences converge to the quantum relative entropy.

In this work, we formulate the AEP for subadditive multipartite entanglement measures. 
The motivation comes from asymptotic entanglement transformations.
Local operations and classical communication (LOCC) channels allow the application of local completely positive maps as well as classical communication between the parties \cite{chitambar2014everything}. A natural question one could ask is whether a state $\psi$ can be transformed into $\varphi$ via LOCC transformation (one-shot setting) \cite{bennett2000exact}. In the asymptotic setting, one aims to find the optimal rate $R$ by which $n$ copies of the pure state $\psi$ can be transformed into $nR + o(n)$ copies of the pure output state $\varphi$ with probability 1. 

This setting can be relaxed by allowing a non-zero probability for failure \cite{bennett2000exact}. In particular, one could ask whether an asymptotic transformation between given pure states is possible with (allowing even vanishing) non-zero probability (SLOCC paradigm \cite{chitambar2008tripartite,yu2010tensor}), or with exponentially decaying failure probability (direct regime), or exponentially decaying success probability (strong converse regime) with a given error-exponent $r$. In the bipartite case, the relation between the transformation rate and error exponents was described in \cite{hayashi2002error}.

The characterization of the optimal rate in the strong converse regime was derived in \cite{jensen2019asymptotic}, building on the work of Strassen \cite{strassen1988asymptotic} in the context of tensors. They showed that the optimal rate is
\begin{equation}\label{eq:Rconverse}
R(\ket{\psi}\to\ket{\varphi},r)=\inf_{F \in \Delta(S_k)} \frac{r \alpha(F) + \log F(\psi)}{\log F(\varphi)},
\end{equation}
where $\Delta_k$ is the asymptotic spectrum of LOCC, i.e., the set of functionals $F\in\Delta_k$ mapping from $k$-partite vectors to $\reals$ that are invariant under local isometries, scale as $F(\sqrt{p}\psi)=p^\alpha F(\psi)$ for some $\alpha\in [0,1]$, are normalized on the $r$-level GHZ state (unit tensor) to $r$, multiplicative under tensor product, additive under tensor sum and monotone under LOCC transformations, i.e., if $\psi\LOCC\varphi$ then $F(\psi)\ge F(\varphi)$.

Although we have a characterization of the optimal achievable rates in terms of the entanglement measures that constitute $\Delta_k$, their explicit form is known only in the bipartite case \cite{hayashi2002error, jensen2019asymptotic}. 
These are exactly the functions $2^{(1-\alpha)\entropy_\alpha(\Tr_1\vectorstate{\psi})}$ where $\entropy_\alpha(\Tr_1\vectorstate{\psi})$ is the $\alpha$-R\'enyi entropy of the Schmidt coefficients of the bipartite pure state $\vectorstate{\psi}$. 
When $k\ge 3$, the explicit characterization of the spectrum is not known, partial results and bounds for the optimal transformation rates were provided in \cite{bugar2022interpolating,bugar2024explicit,bugar2025error}. In \cite{vrana2023family} a family of spectrum-elements parametrized by $\alpha\in [0,1]$ and a probability distribution $\theta\in \distributions([k])$ was constructed.
We will address these functionals in more detail in \cref{sub:eval}.

Another way to relax the asymptotic LOCC setting is given by allowing the fidelity between the target and the outcome states to (vanishingly) deviate from 1, while requiring an asymptotically vanishing error. If we even allow a sublinear amount of quantum communication between the parties, we recover $\text{LOCC}_q$ reducibility \cite{bennett2000exact}.

The optimal rate in the $\text{LOCC}_q$ paradigm was characterized in \cite{vrana2022asymptotic}, in a similar manner as the converse rate, as follows.
\begin{equation}\label{eq:RLOCCq}
R_{\text{LOCC}_{q}} \left( \ket{\psi} \to \ket{\varphi} \right) = \inf_{\substack{E\in\mathcal{F}_k \\ E({\varphi}) \neq 0}} \frac{E({\psi})}{E({\varphi})},
\end{equation}
where the infimum is taken over functionals $E\in\mathcal{F}_k$ defined on $k$-partite unit vectors (on arbitrary $k$-partite Hilbert spaces) such that $E$ is invariant under local isometries and
\begin{enumerate}[({F}1)]
    \item\label{it:norm}
    they are normalized on the $\GHZ$ state to $1$,
    \item \label{it:additive}
    {\it fully additive}, i.e., $E(\psi\otimes\varphi)=E(\psi)+E(\varphi)$
    \item\label{it:monotone}
    {\it monotone on average under LOCC}, i.e., assuming
    \begin{equation*}
\vectorstate{\psi}
\overset{\text{LOCC}}{\longrightarrow} 
\sum_{x \in \mathcal{X}} P(x) 
\lvert \varphi_x \rangle \langle \psi_x \rvert \otimes \lvert x \rangle \langle x \rvert,
\end{equation*}
where $P$ is a probability distribution over $\mathcal{X}$ and $x\in\mathcal{X}$ are the possible values of the classical register, we have 
\begin{equation*}
    E({\psi})\ge \sum_{x \in \mathcal{X}} P(x) 
E({\psi_x}),
\end{equation*}
\item\label{it:asymptoticCont}
asymptotically continuous (\cref{def:asymptcont}). By \cite{vrana2022asymptotic} this is equivalent to ask for 
\begin{equation*}
    E\left( \sqrt{p} \varphi \oplus \sqrt{1-p} \psi \right) = p E(\varphi) + (1-p) E(\psi) + h(p)
\end{equation*}
to hold for any $p\in [0,1]$ and pure $k$-partite vectors $\psi$, $\varphi$, if the other axioms are satisfied.
\end{enumerate}
Note that this optimization resembles \cref{eq:Rconverse}, but the role of logarithmic spectrum elements $\log F$ are taken by the functionals $E$ and $r$ is set to zero. 
In the bipartite case the only functional satisfying these conditions is the entanglement entropy $\entropy(\Tr_j\vectorstate{\psi})$, and the known transformation rate $\frac{\entropy(\Tr_j\vectorstate{\psi})}{\entropy(\Tr_j\vectorstate{\varphi})}$ \cite{bennett1996concentrating,thapliyal2003multipartite} is recovered. 

The fact that in the bipartite case the R\'enyi entropies characterizing the strong converse regime, are reduced to the entanglement entropy when considering $\text{LOCC}_q$ transformability, points to the direction that these settings may be connected even in the multipartite case via some generalization of the asymptotic equipartition property. 
In this paper we explore this connection as follows. After a short preliminary section, in \cref{sec:smoothedSpectrum} we formulate the asymptotic equipartition property (AEP) for subadditive multipartite entanglement measures by defining their smoothing and showing its properties.
In \cref{sec:AEPforknown} we calculate the smoothing limit of the family of spectrum elements introduced in \cite{vrana2023family}. These calculations yield the asymptotically continuous entanglement measures
\begin{equation}
    E^{\theta}(\psi)= \sum_{j\in [k]} \theta_j \entropy(\Tr_j\vectorstate{\psi}),
\end{equation}
where $\entropy(\Tr_j\vectorstate{\psi})$ is the von Neumann entropy of $\Tr_j\vectorstate{\psi}$ (called entanglement entropy), and $\{\theta_j\}_{j\in [k]}$ is a probability distribution over the subsystems.

\section{Notations and Preliminaries}

Throughout this paper, $\log$ is understood as a base $2$ logarithm.
We denote the set of integers from $1$ to $k\in\naturals$ by $[k]$.
The set of unit vectors in the Hilbert space $\Hilbert$ are denoted by $\unitvectors(\Hilbert)$.
In this paper we consider maps $E$, sometimes referenced as entanglement measures, defined over $k$-partite Hilbert spaces of arbitrary dimension. This is formally written as
\begin{equation}
        E : \bigcup_{d_1, \ldots, d_k=1}^\infty (\mathbb{C}^{d_1} \otimes \cdots \otimes \mathbb{C}^{d_k}) \to \mathbb{R}.
\end{equation}
These maps are well defined on tensor products and tensor sums in the following way. For the $k$-partite vectors $\psi\in\Hilbert_1\otimes\dots\otimes\Hilbert_k$ and $\varphi\in\mathcal{K}_1\otimes\dots\otimes\mathcal{K}_k$ their tensor product also forms a $k$-partite vector $\psi\otimes\varphi\in(\Hilbert_1\otimes\mathcal{K}_1)\otimes\dots\otimes(\Hilbert_k\otimes\mathcal{K}_k)$. Similarly, the tensor sum forms the $k$-partite vector $\psi\oplus\varphi\in\left(\Hilbert_1\otimes\dots\otimes \Hilbert_k\right)\oplus \left(\mathcal{K}_1\otimes\dots\otimes \mathcal{K}_k\right)\subseteq(\Hilbert_1\oplus \mathcal{K}_1)\otimes\dots\otimes(\Hilbert_k\oplus \mathcal{K}_k)$. We always treat tensor sums as elements of the latter, larger Hilbert space, because it possesses a $k$-partite structure.

For the states $\rho,\sigma\in\states(\Hilbert)$ we write $\rho\LOCC\sigma$ iff there is an LOCC channel $\Lambda$ such that $\Lambda(\rho)=\sigma$. When we want to emphasize the map $\Lambda$, we write $\rho\xrightarrow{\Lambda}\sigma$.
We will denote the set of unit vectors in the Hilbert space $\Hilbert$ by $\unitvectors(\Hilbert)$.
We associate unit vectors $\psi,\varphi\in\unitvectors(\Hilbert)$ with the pure states $\vectorstate{\psi},\vectorstate{\varphi}\in\states(\Hilbert)$, and write $\psi\LOCC\varphi$ when the associated states satisfy the corresponding relation. We also allow trace non-increasing transformations, addressing cases when the transformation succeeds with probability less than $1$. Formally, for any (not necessarily unit) vectors $\psi,\varphi\in\Hilbert$, $\ket{\psi}\LOCC\ket{\varphi}$ means that ${\textstyle \frac{\ket{\psi}}{\norm{\psi}}}$ can be transformed into ${\textstyle \frac{\ket{\varphi}}{\norm{\varphi}}}$ with success probability ${\textstyle \frac{\norm{\varphi}^2}{\norm{\psi}^2}}$, using a LOCC channel.

\subsection{R\'enyi entropies and statistical distances}

In the following we introduce some quantities from classical and quantum information theory, which will be used later in this work. These concepts can be found in the books \cite{csiszar2011information,cover2012elements,wilde2013quantum,tomamichel2015quantum,nielsen2010quantum}.
Let $\mathcal{X}$ be a finite set, and let $\{P(x)\}_{x\in\mathcal{X}}$ be a probability distribution over $\mathcal{X}$. For $\alpha>0$ and $\alpha\neq1$, the R\'enyi entropy \cite{renyi1961measures} of order $\alpha$ is defined as
\begin{equation}
    \entropy_\alpha(P)\coloneqq\frac{1}{1-\alpha}\log \sum_{x\in\mathcal{X}}P(x)^\alpha.
\end{equation}
The family of R\'enyi entropies is extended to $\alpha=1$ by
\begin{equation}
    \entropy_1(P)\coloneqq\lim_{\alpha\to1}\entropy_\alpha(P)=\entropy(P)=-\sum_{x\in\mathcal{X}}P(x)\log P(x),
\end{equation}
where we recover the Shannon entropy.
$\entropy_\alpha(P)$ is decreasing as a function of $\alpha$, which makes its limiting cases, the min entropy $H_\infty(P)\coloneqq\lim_{\alpha\to\infty}\entropy_\alpha(P)=-\log\max_{x\in\mathcal{X}}P(x)$ and the max entropy $H_0(P)\coloneqq\lim_{\alpha\to0}\entropy_\alpha(P)=\log\left|x\in\mathcal{X}: P(x)>0 \right|$ be the minimal and maximal elements of the one parameter family of R\'enyi entropies.
The R\'enyi entropies are nonnegative and equal to zero iff the distribution $P$ is the Dirac distribution. Their maximal value $\log\lvert\mathcal{X}\rvert$ is attained on the uniform distribution.
The R\'enyi entropies admit a variational characterization \cite{arikan1996inequality,merhav1999shannon,shayevitz2011renyi} as follows. Let $\alpha\in(0,1)$ and $P\in\distributions(\mathcal{X})$, then
\begin{equation}\label{eq:variationalRenyientropy}
\entropy_\alpha(P)=\max_{Q\in\distributions(\mathcal{X})}\left[\entropy(Q)-\frac{\alpha}{1-\alpha}\relativeentropy{Q}{P}\right].
\end{equation}
where $\relativeentropy{.}{.}$ is the Kullback–Leibler divergence (relative-entropy) of the distributions $P,Q\in\distributions(\mathcal{X})$, defined as
\begin{equation}
    \relativeentropy{P}{Q}  \coloneqq \sum_{x\in\mathcal{X}}P(x)\log\frac{P(x)}{Q(x)}
    =-\entropy(P)-\sum_{x\in\mathcal{X}}P(x)\log Q(x)
    .
\end{equation}
The Kullback–Leibler divergence is non-negative, and equals to zero iff $P=Q$. When $\lvert\mathcal{X}\rvert=2$, we use the notation
$h(p)=-p\log p -(1-p)\log (1-p)$ for the binary (Shannon) entropy, and $d(p,q)=p\log\frac{p}{q}+(1-p)\log\frac{1-p}{1-q}$ for the binary relative entropy.

Another measure of similarity between probability distributions is given by the total variation distance defined as
\begin{equation}
    \TV(P,Q)\coloneqq \frac{1}{2}\sum_{x\in\mathcal{X}}\lvert P(x)-Q(x) \rvert.
\end{equation}

In quantum information theory, instead of probability distributions, we mainly work with positive semidefinite unit trace operators over a Hilbert space $\rho, \sigma\in\states(\Hilbert)$, called states. The analogue of total variation distance for states is the trace distance
\begin{equation}
    \T(\rho,\sigma)\coloneqq\frac{1}{2}\norm[1]{\rho-\sigma}=\frac{1}{2}\Tr\lvert \rho-\sigma\rvert.
\end{equation}
When $\rho$ and $\sigma$ are diagonal, the trace distance reduces to the total variation distance.
The trace distance admits the data processing inequality
\begin{equation}
    \T(\rho,\sigma) \ge \T(\Lambda(\rho),\Lambda(\sigma)),
\end{equation}
where $\Lambda:\boundeds(\Hilbert)\to\boundeds(\mathcal{K})$ is a quantum channel, i.e., a completely positive trace preserving map.
Such channels include the partial trace and any LOCC transformation that keeps all the outcomes (otherwise it is trace non-increasing).
Considering pure states, the trace distance takes the form
\begin{equation}\label{eq:tracedistofpure}
    \T(\vectorstate{\varphi},\vectorstate{\psi})=\sqrt{1-\lvert\braket{\varphi}{\psi}\rvert^2}.
\end{equation}

\subsection{LOCC transformations}\label{sub:LOCC}

In this work, we follow the formulation of LOCC protocols given in \cite{jensen2019asymptotic}, which we summarize briefly here.
In this framework, states can be seen as the positive elements  $\rho\in\boundeds(\Hilbert_1)\otimes\dots\otimes\boundeds(\Hilbert_k)\otimes\Diag(\complexes^\mathcal{X})$, where $\mathcal{X}$ is a finite set and $\Diag(\complexes^{\mathcal{X}})$ is the space of diagonal matrices over $\complexes^{\mathcal{X}}$. This additional space serves as the classical register attached to the quantum system (available for all local parties). Note that normalization is not required, as mentioned before, to any unnormalized state $\rho$ we associate the normalized state $\rho/\Tr\rho$, while the norms will correspond to success probabilities between transformations.  

We call
\begin{equation}
    \Lambda : \rho \mapsto \sum_{y \in \mathcal{Y}} \left( (K_y)_i \otimes |y\rangle \langle f(y)| \right) \rho \left( (K_y^*)_i \otimes |f(y)\rangle \langle y| \right)
\end{equation}
a remembering one-step LOCC channel mapping from $\boundeds(\Hilbert_1)\otimes\dots\otimes\boundeds(\Hilbert_k)\otimes\Diag(\complexes^\mathcal{X})$ to $\boundeds(\Hilbert_1)\otimes\dots\otimes\boundeds(\mathcal{K}_i)\otimes\dots\otimes\boundeds(\Hilbert_k)\otimes\Diag(\complexes^{\mathcal{Y}})$, where $\mathcal{Y}$ is a finite index set, $f:\mathcal{Y}\mapsto\mathcal{X}$ is a map, and for each $y\in\mathcal{Y}$, $K_y:\Hilbert_y\mapsto\mathcal{K}_y$ is a Kraus operator (linear map satisfying $\sum_{y\in\mathcal{Y}}K_y^*K_y\otimes\vectorstate{f(y)}\le I_{\Hilbert_i\otimes\complexes^{\mathcal{Y}}}$) and $(K_y)_i=I_{\Hilbert_1}\otimes\dots\otimes K_y \otimes\dots\otimes I_{\Hilbert_k}$. 
Intuitively this means, that party $i$ reads out the classical register, which determines the local channel they use to produce the new state, including the new register value. 

In \cite[Prop. 2.13.]{jensen2019asymptotic} they had shown that any LOCC channel $\Lambda$ (between states without the classical register), can be written as successive application of remembering one-step LOCC channels $(\Lambda_1,\dots,\Lambda_n)$ followed by the partial trace on the classical register at the end of the protocol, i.e.
\begin{equation}
    \Lambda=\Tr_{\Diag(\complexes^{\mathcal{X}})} \circ\Lambda_n\circ\dots\circ\Lambda_1.
\end{equation}
Then it follows, that the outcome of any LOCC channel, when applied to a pure state, can be written in the form  
\begin{equation}\label{eq:conditionallyPure}
    \Lambda(\vectorstate{\psi}) = \Tr_{\Diag(\complexes^{\mathcal{X}})}\sum_{x\in\mathcal{X}} P(x)\vectorstate{\psi_x}\otimes\vectorstate{x}
    =\sum_{x\in\mathcal{X}} P(x)\vectorstate{\psi_x}
    .
\end{equation}
Positive operators in this form (after $\Tr_{\Diag(\complexes^{\mathcal{X}})}$) are called {\it conditionally pure states}.

\section{Asymptotic equipartition property of entanglement measures}\label{sec:smoothedSpectrum}

Inspired by the entropic formulation of the classical AEP, in the following we define the smoothing map, mapping between multipartite entanglement measures. 

\begin{definition}\label{def:smoothing}
Let $E$ be a map from the non-zero $k$-partite unit vectors to $\nonnegativereals$.  
For $\epsilon\in(0,1]$ let
\begin{equation}\label{eq:smoothedSpec}
         E^\epsilon(\psi)\coloneqq  \inf_{\substack{\varphi \in \unitvectors(\mathcal{H})\\ \lvert\langle \varphi \rvert\psi \rangle|^2 \geq 1 - \epsilon}} E(\varphi),
\end{equation}
and  
\begin{equation}\label{eq:smoothedSpec2}
        \Phi(E)(\psi)\coloneqq\lim_{\epsilon\to 0}\limsup_{n \to \infty} \frac{1}{n} E^\epsilon(\psi^{\otimes n})=\lim_{\epsilon\to0}\limsup_{n \to \infty} \frac{1}{n}
        \inf_{\substack{\varphi \in \unitvectors(\mathcal{H}^{\otimes n})\\ \lvert\langle \varphi \rvert\psi^{\otimes n} \rangle|^2 \geq 1-\epsilon}} E(\varphi)
\end{equation}
be the {\it smoothing limit},
where the existence of the limit is the consequence of the monotonicity of $E^\epsilon(\psi^{\otimes n})$ in $\epsilon$, and that $E^\epsilon(\psi)\le E(\psi)$.

\end{definition}

\begin{proposition}\label{prop:subadditivityKeeping}
    Let $E$ be a map from the non-zero $k$-partite unit vectors to $\nonnegativereals$. Then $\Phi(E)$ is weakly additive, i.e. $\Phi(E)(\psi_1^{\otimes n})=n\Phi(E)(\psi_1)$, for any $k$-partite unit vector $\psi_1$ and $n\in\integers$. Moreover if $E$ is subadditive under the tensor product, then $\Phi(E)$ is also subadditive, i.e. $\Phi(E)(\psi_1\otimes\psi_2)\le \Phi(E)(\psi_1)+\Phi(E)(\psi_2)$ for any $k$-partite unit vector $\psi_1, \psi_2$.
    \end{proposition}
\begin{proof}
The weak additivity is the direct consequence of the regularization in the definition, namely that
\begin{equation}
    \limsup_{n \to \infty} \frac{1}{n}\Phi_{\epsilon} (E)(\psi^{\otimes mn})=
    \limsup_{n \to \infty} \frac{m}{n}\Phi_{\epsilon} (E)(\psi^{\otimes n}).
\end{equation}

The non-negativity of $\Phi(E)$ follows trivially from the definition of $\Phi$ and the non-negativity of $E$. 
Let $\psi=\psi_1\otimes\psi_2$, $\psi_1\in\unitvectors(\Hilbert)$ and $\psi_2\in\unitvectors(\mathcal{K})$ be unit vectors, where $\Hilbert=\Hilbert_1\otimes\dots\otimes\Hilbert_k$ and $\mathcal{K}=\mathcal{K}_1\otimes\dots\otimes\mathcal{K}_k$. By considering $\varphi=\varphi_1\otimes\varphi_2$ where $\varphi_1\in\unitvectors(\Hilbert^{\otimes n})$ and $\varphi_2\in\unitvectors(\mathcal{K}^{\otimes n})$, and
noticing that $\lvert\langle \varphi_1\rvert\psi_1^{\otimes n}
    \rangle\rvert^2\ge \sqrt{1-\epsilon}$ together with $\lvert\langle \varphi_2 \rvert\psi_2^{\otimes n}
    \rangle\rvert^2 \ge \sqrt{1-\epsilon}$ implies $\lvert\langle \varphi\rvert\psi_1^{\otimes n}
    \otimes
    \psi_2^{\otimes n}
    \rangle\rvert^2 \geq 1-\epsilon$,
we write
\begin{equation}
\begin{split}
    & \inf_{\substack{\varphi \in \unitvectors(\mathcal{H}^{\otimes n}\otimes\mathcal{K}^{\otimes n})\\ \lvert\langle \varphi\rvert\psi_1^{\otimes n}
    \otimes
    \psi_2^{\otimes n}
    \rangle\rvert^2 \geq 1-\epsilon}} E(\varphi)\le \\
    &\inf_{\substack{\varphi_1 \in \unitvectors(\mathcal{H}^{\otimes n})\\ \lvert\langle \varphi_1\rvert\psi_1^{\otimes n}
    \rangle\rvert^2 \geq \sqrt{1-\epsilon}}} \qquad
    \inf_{\substack{\varphi_2 \in \unitvectors(\mathcal{K}^{\otimes n})\\ \lvert\langle \varphi_2 \rvert\psi_2^{\otimes n}
    \rangle\rvert^2 \geq \sqrt{1-\epsilon}}} 
    E(\varphi_1\otimes\varphi_2)\le\\
    &\inf_{\substack{\varphi_1 \in \unitvectors(\mathcal{H}^{\otimes n})\\ \lvert\langle \varphi_1\rvert\psi_1^{\otimes n}
    \rangle\rvert^2 \geq \sqrt{1 - \epsilon}}}
    E(\varphi_1)+
    \inf_{\substack{\varphi_2 \in \unitvectors(\mathcal{K}^{\otimes n})\\ \lvert\langle \varphi_2 \rvert\psi_2^{\otimes n}
    \rangle\rvert^2 \geq \sqrt{1 - \epsilon}}} 
    E(\varphi_2)\le\\
    &\inf_{\substack{\varphi_1 \in \unitvectors(\mathcal{H}^{\otimes n})\\ \lvert\langle \varphi_1\rvert\psi_1^{\otimes n}
    \rangle\rvert^2 \geq 1 - \frac{1}{2}\epsilon}}
    E(\varphi_1)+
    \inf_{\substack{\varphi_2 \in \unitvectors(\mathcal{K}^{\otimes n})\\ \lvert\langle \varphi_2 \rvert\psi_2^{\otimes n}
    \rangle\rvert^2 \geq 1 - \frac{1}{2}\epsilon}} 
    E(\varphi_2)
\end{split}
\end{equation}
In the second inequality we used the subadditivity of $E$, then the bound $\sqrt{1-\epsilon}\le1-\frac{1}{2}\epsilon$. We finish this by dividing the inequality by $n$ and taking the limsup $n\to\infty$, then the limit $\epsilon \to 0$.  

\end{proof}

\begin{definition}\label{def:deltasub}
    We denote the set of those maps
    $E$ mapping $k$-partite unit vectors into $\nonnegativereals$ by $\mathcal{F}_{\text{sub},k}$, which are subadditive, i.e. $E(\psi_1\otimes \psi_2)\le E(\psi_1)+E(\psi_2)$; logarithmically bounded, i.e.
    \begin{equation}\label{eq:boundedSum}
    E\left({\psi_1}\oplus\dots\oplus{\psi_l}\right)\le
    \max_{i\in[l]} E\left(\frac{\psi_i}{\norm{\psi_i}}\right) + \log l,
\end{equation}
where $\psi_i$ are such that $\psi_1\oplus\dots\oplus\psi_l$ is a unit vector; and monotone on average \labelcref{it:monotone}.
    
\end{definition}

\begin{proposition}
    $\mathcal{F}_k\subset \mathcal{F}_{\text{sub},k}$
    \end{proposition}
    \begin{proof}
    Let $E\in\mathcal{F}_k$.
    To show \cref{eq:boundedSum} we write
    \begin{equation}
    \begin{split}
        E\left(\psi_1\oplus\dots\oplus\psi_l\right) &=
        \sum_{i\in[l]} \norm{\psi_i}^2 E\left(\frac{\psi_i}{\norm{\psi_i}}\right)+ \entropy(\{\norm{\psi_i}^2\}_{i\in[l]})\\
        &\le
        \max_{i\in[l]}E\left(\frac{\psi_i}{\norm{\psi_i}}\right) + \log l
        \end{split}
    \end{equation}
    by successively using \labelcref{it:asymptoticCont} along with the chain rule of Shannon entropy to get the first inequality. The rest is true by assumption.
\end{proof}

\begin{lemma}\label{lem:weaksuperadd}
    Let $E\in\mathcal{F}_{\text{sub},k}$. Then
    \begin{equation}
        \Phi(E)(\psi_1\otimes\psi_2)\ge \Phi(E)(\psi_1)
    \end{equation}
    for $\psi_1\in\Hilbert$ and $\psi_2\in\mathcal{K}$.
\end{lemma}
\begin{proof}
Let $\varphi\in\unitvectors(\Hilbert_1^{\otimes n}\otimes\Hilbert_2^{\otimes n})$ be such that $\lvert\braket{\varphi}{\left(\psi_1\otimes\psi_2\right)^{\otimes n}}\rvert^2\ge 1-\epsilon$
    We apply a measurement channel $\mathcal{M}$ on $\Hilbert_2^{\otimes n}$, which by \cref{eq:conditionallyPure} transforms any vector as
    \begin{equation}
        \mathcal{M}(\vectorstate{\varphi})=
        \sum_{x\in\mathcal{X}} Q(x)\vectorstate{\varphi_x}\otimes\vectorstate{x},
    \end{equation}
where $\varphi_x\in\Hilbert_1^{\otimes n}$ for each outcome $x\in\mathcal{X}$.
In particular $(\psi_1\otimes\psi_2)^{\otimes n}$ is transformed into $\psi_1^{\otimes n}$ by this channel (up to the classical register). By \cref{lem:LOCCoutcomeproblower} and the monotonicity on average we get
\begin{equation}
    \begin{split}
    E(\varphi)&\ge
    \sum_{x\in\mathcal{X}} Q(x)
    E(\varphi_x)\\
    &\ge
        \sum_{\substack{x\in\mathcal{X}\\\lvert\braket{\varphi_x}{\psi_1^{\otimes n}}\rvert^2\ge 1-\epsilon'}} 
     Q(x)
    E(\varphi_x)\\
    &\ge
        \sum_{\substack{x\in\mathcal{X}\\\lvert\braket{\varphi_x}{\psi_1^{\otimes n}}\rvert^2\ge 1-\epsilon'}} 
     Q(x)
     \inf_{\substack{\varphi \in \unitvectors(\Hilbert^{\otimes n})\\ \lvert\langle \varphi \rvert\psi_1^{\otimes n} \rangle|^2 \geq 1 - \epsilon'}}
    E(\varphi)\\
    &\ge
    \left(1
     -2\sqrt{\epsilon}\left(1+\frac{1}{\sqrt{\epsilon'}}\right)\right)
     \inf_{\substack{\varphi \in \unitvectors(\Hilbert^{\otimes n})\\ \lvert\langle \varphi \rvert\psi_1^{\otimes n} \rangle|^2 \geq 1 - \epsilon'}}
    E(\varphi)\\
    .
\end{split}
\end{equation}
The proof is finished by applying \cref{def:smoothing}, i.e. taking the infimum over $\varphi$ on the left hand side then we divide the inequality by $n$, take the $\limsup$ in $n$ and the limit $\epsilon\to 0$. At the end we also take limit $\epsilon'\to 0$.

\end{proof}

\begin{lemma}\label{prop:PhiEproperties}
    Let $E\in\mathcal{F}_{\text{sub},k}$.
    Then
    \begin{equation}\label{eq:asympContle}
    \Phi(E)\left( \sqrt{p} \varphi \oplus \sqrt{1-p} \psi \right) \le p \Phi(E)(\varphi) + (1-p) \Phi(E)(\psi) + h(p)
\end{equation}
for any $p\in [0,1]$.
\end{lemma}

\begin{proof}
Let $\psi=\sqrt{p}\psi_1\oplus \sqrt{1-p}\psi_2 \in \unitvectors(\Hilbert\oplus\mathcal{K})$ for given $0<p<1$, where $\psi_1\in\unitvectors(\Hilbert)$ and $\psi_2\in\unitvectors(\mathcal{K})$.
Its $n$-th tensor power can be decomposed into the binomial series 
\begin{equation}\label{eq:psitensorndecomp}
    \psi^{\otimes n}=\bigoplus_{m=1}^n \bigoplus_{i=1}^{n\choose m} \sqrt{p}^{\, m}\psi_1^{\otimes m}\otimes\sqrt{1-p}^{\, n-m}\psi_2^{\otimes n-m}.
\end{equation}
Note that the norm squares of the tensor summands form a binomial distribution with parameters $n, p$.
By the weak law of large numbers for any $\delta>0$ and $\epsilon>0$ one can choose $n$ large enough so that
\begin{equation}
\sum_{m=\lfloor n(p-\delta)\rfloor}^{\lfloor n(p+\delta)\rfloor} 
{n \choose m} p^{m} (1-p)^{n-m}
\ge \sqrt[4]{1-\epsilon}.
\end{equation}
Consider the vector
\begin{equation}\label{eq:varphitensorndecomp}
       \varphi\coloneqq\bigoplus_{m=\lfloor n(p-\delta)\rfloor}^{\lfloor n(p+\delta)\rfloor}  \varphi^{(m)}
       \coloneqq\bigoplus_{m=\lfloor n(p-\delta)\rfloor}^{\lfloor n(p+\delta)\rfloor} \bigoplus_{i=1}^{n\choose m}\sqrt{p}^{\, m}\varphi_{1}^{(m)}\otimes\sqrt{1-p}^{\, n-m}\varphi_{2}^{(n-m)},
\end{equation}
where $\varphi_1^{(m)}\in\unitvectors(\Hilbert^{\otimes m}), \varphi_2^{(n-m)}\in\unitvectors(\mathcal{K}^{\otimes n-m})$ are such that $\braket{\varphi_{1}^{(m)}}{\psi_1^{\otimes m}}\ge \sqrt[8]{1-\epsilon}$ and $\braket{\varphi_{2}^{(n-m)}}{\psi_2^{\otimes n-m}}\ge\sqrt[8]{1-\epsilon}$ (the phases of these unit vectors can always be chosen such that the inner products are non-negative). Note that $\norm{\varphi}\le 1$.
We write the bound
\begin{equation}
\begin{split}
\left\lvert\braket{\frac{\varphi}{\norm{\varphi}}}{\psi^{\otimes n}}\right\rvert^2
&=\frac{1}{\norm{\varphi}^2}\left\lvert\sum_{m=\lfloor n(p-\delta)\rfloor}^{\lfloor n(p+\delta)\rfloor} {n \choose m} p^{\, m}\braket{\varphi_{1}^{(m)}}{\psi_1^{\otimes m}} (1-p)^{\, n-m}\braket{\varphi_{2}^{(n-m)}}{\psi_2^{\otimes n-m}}\right\rvert^2\\
&\ge 
\left\lvert
\sqrt[4]{1-\epsilon}
\sum_{m=\lfloor n(p-\delta)\rfloor}^{\lfloor n(p+\delta)\rfloor}
{n \choose m} p^{\, m} (1-p)^{\, n-m}\right\rvert^2\\
&\ge
1-\epsilon
\end{split}
\end{equation}
to see that this choice of $\varphi$ indeed satisfies the conditions of the infimum.
By assumption (logarithmic boundedness) we have
\begin{equation}\label{eq:EasymptContMainUpperB}
\begin{split}
     E(\varphi)
    \le
     & 
    \max_{m\in [\lfloor n(p-\delta)\rfloor,\lfloor n(p+\delta)\rfloor]}
    \quad
     E\left(\varphi_{1}^{(m)}\otimes\varphi_{2}^{(n-m)}\right)
     + \log \left(\lfloor 2n\delta\rfloor+1\right)
     +\log{n \choose m}.
\end{split}
\end{equation}
Let $m_n$ be the maximizing value of $m$ for given $n$. 
By subadditivity and the bound ${n \choose m_n}\le 2^{nh(\frac{m_n}{n})}$ (\cite{csiszar2011information}) we have 
\begin{equation}
\begin{split}
     E(\varphi)
    \le
     & 
     E\left(\varphi_{1}^{(m_n)}\otimes\varphi_{2}^{(n-m_n)}\right)
     + \log \left(\lfloor 2n\delta\rfloor+1\right)
     +\log{n \choose m_n}\\
     \le&
     E\left(\varphi_{1}^{(m_n)}\right)
     + E\left(\varphi_{2}^{(n-m_n)}\right)
     +  \log \left(\lfloor 2n\delta\rfloor+1\right)
     +nh\left(\frac{m_n}{n}\right)\\
\end{split}
\end{equation}

Let $p_\delta\coloneqq\limsup_{n\to\infty} \frac{m_n}{n}$, which by definition satisfies $p-\delta\le p_\delta\le p+\delta$.
Taking the infimum of \cref{eq:EasymptContMainUpperB} only over vectors $\varphi$ of the form \cref{eq:varphitensorndecomp} with the corresponding conditions, we get an upper bound on the infimum in the smoothing. Then dividing by $n$ and taking the limsup $n\to\infty$, the term ${\textstyle \frac{1}{n}\log \lceil 2n\delta\rceil}$ vanishes, and we get
\begin{equation}\label{eq:asContmainupperbound}
\begin{split}
\limsup_{n\to\infty}
     \frac{1}{n} \inf_{\substack{\varphi \in \unitvectors((\mathcal{H}\oplus\mathcal{K})^{\otimes n})\\ \lvert\langle \varphi \rvert\psi^{\otimes n} \rangle| \geq 1-\epsilon}} E(\varphi)
    \le
     &
     h(p_\delta)+ 
    \limsup_{m\to\infty}\frac{p_\delta}{m}\inf_{\substack{\varphi_{1,m}\\
    \braket{\varphi_{1,m}}{\psi_1^{\otimes m}}\ge \sqrt[4]{1-\epsilon}
    }} E(\varphi_{1,m}) 
    \\&+ \limsup_{m\to\infty}\frac{1-p_\delta}{m}\inf_{\substack{\varphi_{2,m}\\
    \braket{\varphi_{2,m}}{\psi_2^{\otimes m}}\ge \sqrt[4]{1-\epsilon}
    }} E(\varphi_{2,m})\\
    \le
    &
     h(p_\delta)+ 
     p_\delta\limsup_{m\to\infty}\frac{1}{m}\inf_{\substack{\varphi_{1,m}\\
    \braket{\varphi_{1,m}}{\psi_1^{\otimes m}}\ge 1-\frac{1}{4}\epsilon
    }} E(\varphi_{1,m_n}) 
    \\&+ (1-p_\delta)\limsup_{m\to\infty}\frac{1}{m}\inf_{\substack{\varphi_{2,m}\\
    \braket{\varphi_{2,m}}{\psi_2^{\otimes m}}\ge 1-\frac{1}{4}\epsilon
    }} E(\varphi_{2,m}).
\end{split}
\end{equation}
In the last inequality we use $\sqrt[4]{1-\epsilon}\le 1-\frac{1}{4}\epsilon$.
Taking the limit $\delta\to 0$, we get $p_\delta
\to p$. Then $\epsilon\to 0$ gives the desired inequality.
\end{proof}

\begin{proposition}\label{lem:phiEsubweakmon}
    Let $E\in\mathcal{F}_{\text{sub},k}$. Then $\Phi(E)$ is logarithmically bounded \eqref{eq:boundedSum}, and weakly monotone, i.e. if there exists an LOCC protocol
\begin{equation*}
\vectorstate{\psi}
\overset{\text{LOCC}}{\longrightarrow} 
\sum_{x \in \mathcal{X}} P(x) 
\lvert \varphi_x \rangle \langle \varphi_x \rvert \otimes \lvert x \rangle \langle x \rvert,
\end{equation*}
where $P$ is a probability distribution over the possible values of the classical register $\mathcal{X}$, then we have 
\begin{equation*}
    E({\psi})\ge  P(x)E({\varphi_x})
\end{equation*}
for each $x\in\mathcal{X}$.
Furthermore if $\Phi(E)$ is additive, then it is also monotone on average \labelcref{it:monotone}.
\end{proposition}
\begin{proof}

 We write
       \begin{equation}
    \begin{split}
        \Phi(E)\left(\psi_1\oplus\dots\oplus\psi_l\right) &\le
        \sum_{i\in[l]} \norm{\psi_i}^2 \Phi(E)\left(\frac{\psi_i}{\norm{\psi_i}}\right)+ \entropy(\{\norm{\psi_i}^2\}_{i\in[l]})\\
        &\le
        \max_{i\in[l]}\Phi(E)\left(\frac{\psi_i}{\norm{\psi_i}}\right) + \log l
        \end{split}
    \end{equation}
    by successively using \cref{prop:PhiEproperties} along with the chain rule of Shannon entropy to get the first inequality. 
    
    We continue with monotonicity (both).
    Assume that there is an LOCC channel {$\Lambda:\boundeds(\Hilbert_1)\otimes\dots\otimes\boundeds(\Hilbert_k)\mapsto\boundeds(\mathcal{K}_1)\otimes\dots\otimes\boundeds(\mathcal{K}_k)\otimes\Diag(\complexes^{\mathcal{X}})$} such that
\begin{equation}
\vectorstate{\psi}
\overset{\Lambda}{\longrightarrow} 
\sum_{x \in \mathcal{X}} P(x) 
\vectorstate{\psi_x}  \otimes \vectorstate{x},
\end{equation}
where $\mathcal{X}$ is a finite set encoding the possible values of the classical register, and the $P(x)\ge 0$ values form a probability distribution (possibly featuring zero probabilities for some $x\in\mathcal{X}$). 
In the following we use the notation $\Hilbert\coloneqq\Hilbert_1\otimes\ldots\otimes\Hilbert_k$ and $\mathcal{K}\coloneqq\mathcal{K}_1\otimes\ldots\otimes\mathcal{K}_k$.
The previous LOCC transformation also defines the LOCC transformation $\Lambda^{\otimes n}$ on $n$-copy states as
\begin{equation}\label{eq:psitensornlambdaprotocol}
\begin{split}
    \vectorstate{\psi}^{\otimes n}\overset{\Lambda^{\otimes n}}{\longrightarrow} 
    &\sum_{\overline{x} \in \mathcal{X}^{n}}
{P^{\otimes n}(\overline{x})} \,
 \vectorstate{\psi_{\overline{x}}}\otimes\vectorstate{\overline{x}}\\=
 \sum_{x_1 \in \mathcal{X}}
 \ldots
 &\sum_{x_n \in \mathcal{X}}
 {
P(x_1)\dots P(x_{n})} \,
\vectorstate{\psi_{x_1}}\otimes\ldots\otimes\vectorstate{\psi_{x_{n}}}\otimes\vectorstate{x_1,\ldots, x_n},   
 \end{split}
\end{equation}
where $P^{\otimes n}(\overline{x})$ is the probability that the LOCC protocol $\Lambda^{\otimes n}$ creates the output state $\psi_{\overline{x}}= \psi_{x_1}\otimes\dots\otimes\psi_{x_{n}}\in\unitvectors\left((\mathcal{K}_1\otimes\dots\otimes\mathcal{K}_k)^{\otimes n}\right)$ corresponding to the $n$-string $\overline{x}\coloneqq(x_1,\dots, x_{n})$ of one-copy register values.
By the weak law of large numbers for any $\epsilon,\delta>0$ we can choose $n$ large enough, so that the total probability of strings $\overline{x}$ not containing every letter $x$ at least $\lfloor n\left(P(x)-\delta\right)\rfloor$ times is upper bounded by $\epsilon$.
Also, for a large enough $n$ for each $x$ there exist $p_x\ge P(x)-\delta$ such that $np_x$ is an integer.

Generally, any $n$-copy of the pure state $\varphi\in\left(\Hilbert_1\otimes\dots\otimes\Hilbert_k\right)^{\otimes n}$ is transformed by $\Lambda^{\otimes n}$ in the same manner (see \cref{sub:LOCC}), i.e.
\begin{equation}
    \vectorstate{\varphi}
\overset{\Lambda^{\otimes n}}{\longrightarrow} 
\sum_{\overline{x} \in \mathcal{X}^{n}} Q(\overline{x}) 
\vectorstate{\varphi_{\overline{x}}}  \otimes \vectorstate{\overline{x}},
\end{equation}
where again $Q(\overline{x})\ge0$, $\sum_{\overline{x}\in\mathcal{X}^n}Q(\overline{x})=1$, and $\varphi_{\overline{x}}$ are unit vectors. The probability distribution $Q$ may have a different support as $P$. Assuming $\lvert\langle \varphi \rvert\psi^{\otimes n} \rangle| \geq 1-\epsilon$, by \cref{lem:LOCCoutcomeproblower} and the union bound (applied to intersections) the total probability of having an outcome $\overline{x}\in\mathcal{X}$ for which $\text{T}(\vectorstate{\varphi_{\overline{x}}},\vectorstate{\psi_{\overline{x}}}) \le \sqrt{\epsilon'}$ and which also contains each letter $x$ at least $np_x$ times at the same time is
\begin{equation}
    \sum_{\substack{\overline{x}\in\mathcal{X}_{\text{accept}}\\\text{T}(\vectorstate{\varphi_{\overline{x}}},\vectorstate{\psi_{\overline{x}}}) \le \sqrt{\epsilon'}} }
     Q(\overline{x})
     \ge  
     1 -2\sqrt{\epsilon}\left(1+\frac{1}{\sqrt{\epsilon'}}\right) -\epsilon,
\end{equation}
where $\mathcal{X}_{\text{accept}}\subseteq \mathcal{X}^n$ denotes those strings which contain each letter at least $np_x$ times.

By the previous inequality, the fact that $E(\varphi_{\overline{x}})\ge0$, and that $E$ is monotone on average, we have
\begin{equation}\label{eq:upperBoundByE}
\begin{split}
      E(\varphi)&\ge
      \sum_{\overline{x}\in\mathcal{X}} 
     Q(\overline{x})
     E(\varphi_{\overline{x}})\\
     &\ge
      \sum_{\substack{\overline{x}\in\mathcal{X}_{\text{accept}}\\\text{T}(\vectorstate{\varphi_{\overline{x}}},\vectorstate{\psi_{\overline{x}}}) \le \sqrt{\epsilon'}} } 
     Q(\overline{x})
     E(\varphi_{\overline{x}})\\&\ge
     \sum_{\substack{\overline{x}\in\mathcal{X}_{\text{accept}}\\\text{T}\\      \left(\vectorstate{\varphi_{\overline{x}}},\vectorstate{\psi_{\overline{x}}}\right) \le \sqrt{\epsilon'}} }
     Q(\overline{x})
     \inf_{\substack{\varphi' \in \unitvectors\left(\mathcal{K}^{\otimes n}\right)\\ \text{T}(\vectorstate{\varphi'},\vectorstate{\psi_{\overline{x}}}) \le \sqrt{\epsilon'}}}
     E(\varphi')\\&\ge
     \sum_{\substack{\overline{x}\in\mathcal{X}_{\text{accept}}\\\text{T}\left(\vectorstate{\varphi_{\overline{x}}},\vectorstate{\psi_{\overline{x}}}\right) \le \sqrt{\epsilon'}} }
     Q(\overline{x})
     \inf_{\substack{\varphi' \in \unitvectors\left(\mathcal{K}^{\otimes n}\right)\\ \text{T}\left(\Lambda_{\text{discard}}(\vectorstate{\varphi'}\right),\bigotimes_{x\in\mathcal{X}}\vectorstate{\psi_x}^{\otimes np_x}) \le \sqrt{\epsilon'}}}
     E(\varphi')\\&\ge
      \left(1 -2\sqrt{\epsilon}\left(1+\frac{1}{\sqrt{\epsilon'}}\right) -\epsilon\right)\quad 
     \inf_{\substack{\varphi' \in \unitvectors\left(\mathcal{K}^{\otimes n}\right)\\ \text{T}\left(\vectorstate{\varphi'},\bigotimes_{x\in\mathcal{X}}\vectorstate{\psi_x}^{\otimes np_x}\right) \le \sqrt{\epsilon'}}}
     E(\varphi')
\end{split}
\end{equation}
In the fourth inequality we applied a discarding channel $\Lambda_{\text{discard}}$ which performs measurement (therefore discarding) the extra copies, and keeps only $np_x$ copies of each letter $x$. There we also used the data processing inequality, namely that
\begin{equation*}
    \text{T}\left(\Lambda_{\text{discard}}(\vectorstate{\varphi'}\right),\bigotimes_{x\in\mathcal{X}}\vectorstate{\psi_x}^{\otimes np_x})\le \text{T}\left(\vectorstate{\varphi'},\vectorstate{\psi_{\overline{x}}}\right) \le \sqrt{\epsilon'}.
\end{equation*}
Then what we have is
\begin{equation}
\begin{split}
    \lim_{\epsilon'\to 0}\lim_{\epsilon\to0}\limsup_{n\to\infty}\frac{1}{n}
    \inf_{\substack{\varphi \in \unitvectors(\mathcal{H}^{\otimes n})\\ \left\lvert \braket{\varphi}{\psi^{\otimes n}}\right\rvert^2 \ge 1- \epsilon}}
    E(\varphi)&\ge
    \lim_{\epsilon'\to 0}
    \limsup_{n\to\infty}\frac{1}{n}
    \inf_{\substack{\varphi \\ \left\lvert \braket{\varphi}{\bigotimes_{x\in\mathcal{X}}\psi_x^{\otimes np_x}}\right\rvert^2 \ge 1- \epsilon'}}
     E(\varphi)\\
    &
    \ge
    \lim_{\epsilon'\to 0}
    \limsup_{k\to\infty}\frac{1}{kn}
    \inf_{\substack{\varphi \\ \left\lvert \braket{\varphi}{\bigotimes_{x\in\mathcal{X}}\psi_x^{\otimes knp_x}}\right\rvert^2 \ge 1- \epsilon'}}
     E(\varphi)
     \\
     &
     =\frac{1}{n}
     \Phi(E)\left(\bigotimes_{x\in\mathcal{X}}\psi_x^{\otimes np_x}\right)
     .
\end{split}
\end{equation}

If $\Phi(E)$ is additive we get the monotonicity on average by decomposing the right hand side and taking $\delta\to0$ (and therefore $p_x\to P(x)$), otherwise we get weak monotonicity the same way by \cref{lem:weaksuperadd} and the weak additivity of $\Phi(E)$ (\cref{prop:subadditivityKeeping}). 

\end{proof}

\begin{definition}\label{def:asymptcont}
    A function 
    \begin{equation}
        f : \bigcup_{d_1, \ldots, d_k=1}^\infty \unitvectors(\mathbb{C}^{d_1} \otimes \cdots \otimes \mathbb{C}^{d_k}) \to \mathbb{R}
    \end{equation}
is {\it asymptotically continuous} if
\begin{equation}
    \frac{f}{1 + \log \dim \mathcal{H}}
\end{equation}
is uniformly continuous.
\end{definition}

\begin{remark}\label{rem:asymptoticCont}
    In \cite[Theorem 4.2.]{vrana2022asymptotic} it is shown that for $E\in\mathcal{F}_{\text{sub},k}$ the inequality 
    \begin{equation}\label{eq:EsumUpperbyHalpha2}
    E\left( \sqrt{p} \varphi \oplus \sqrt{1-p} \psi \right) \le p E(\varphi) + (1-p) E(\psi) + h(p)
\end{equation}
implies asymptotic continuity and the continuity bound
\begin{equation}\label{eq:contBound}
|E(\varphi) - E(\psi)| \leq a\left(\sqrt{1-\lvert\braket{\varphi}{\psi}\rvert^2}\right) \log \dim \mathcal{H} + b\left(\sqrt{1-\lvert\braket{\varphi}{\psi}\rvert^2}\right),
\end{equation}
where
\begin{equation}
a(\delta) = \frac{\left(1 + \delta^{\frac{2}{k+1}}\right)^{k+1} - 1 + \delta^2}{1 - \delta^2}
\end{equation}
and
\begin{equation}
b(\delta) = \frac{\left(1 + \delta^{\frac{2}{k+1}}\right)^{k+1}}{1 - \delta^2} \, h\left(\left(1 + \delta^{\frac{2}{k+1}}\right)^{-1}\right).
\end{equation}
In fact, they also assume that $E$ is additive, normalized, monotone on average (instead of weak monotonicity) and that \cref{eq:EsumUpperbyHalpha2} is satisfied with equality, but in the proof of $(iii)\implies(iv)$ and $(iv)\implies(i)$ we can relax these assumptions\footnote{In particular, monotonicity on average and \cref{eq:EsumUpperbyHalpha2} with equality, both of which are used only in \cite[Equation (58)]{vrana2022asymptotic}, where weak monotonicity and the inequality mentioned are sufficient.}. Combining these facts with \cref{lem:phiEsubweakmon} and \cref{prop:PhiEproperties} we find that for $E\in\mathcal{F}_{\text{sub},k}$ the functional $\Phi(E)$ is asymptotically continuous. 

\end{remark}

\begin{theorem}[AEP]\label{thm:AEP}
    The map $\Phi$ (\cref{def:smoothing}) acts idempotently on the set $\mathcal{F}_{\text{sub},k}$, i.e. for $E\in\mathcal{F}_{\text{sub},k}$ we have $\Phi(\Phi(E))=\Phi(E)$, and its image consists of subadditive, asymptotically continuous, weakly additive and weakly monotone functionals. Furthermore, if $\Phi(E)$ is additive, i.e. $\Phi(E)(\psi_1\otimes\psi_2)= \Phi(E)(\psi_1)+\Phi(E)(\psi_2)$, then it is also monotone on average.
\end{theorem}
\begin{proof}
 The inequality $\Phi(\Phi(E))\le\Phi(E)$ is trivial by \cref{def:smoothing}.
 For the converse we use that by \cref{rem:asymptoticCont} $E'\coloneqq\Phi(E)$ is asymptotically continuous and that by \cref{prop:subadditivityKeeping} it is subadditive and weakly additive. Any asymptotically continuous and weakly additive entanglement measure $E'$ satisfies:
 \begin{equation}\label{eq:idempotentE}
 \begin{split}
     \Phi(E')(\psi)&\ge 
     \lim_{\epsilon\to0}\limsup_{n \to \infty}\bigg( \frac{1}{n}
        \inf_{\substack{\varphi \in \unitvectors(\mathcal{H}^{\otimes n})\\ \lvert\langle \varphi \rvert\psi^{\otimes n} \rangle|^2 \geq 1-\epsilon}} E'(\psi^{\otimes n}) \\&\quad- 
        (1+n\log\dim\Hilbert)c(\epsilon)\bigg)
        \\
        &=
        \lim_{\epsilon\to0}\big(
         E'(\psi) - 
          c(\epsilon)\log\dim\mathcal{H}\big) \\
        &=E'(\psi),
\end{split}
 \end{equation}
where $c(\epsilon)>0$ is a positive valued function of $\epsilon>0$ such that $\lim_{\epsilon\to 0}c(\epsilon)=0$.
The fact that the additive elements of the image are monotone on average follows from \cref{lem:phiEsubweakmon}.
\end{proof}

\begin{corollary}
    Any entanglement measure $E$ in the image of $\Phi$ is monotone under approximate asymptotic transformations, i.e. if $\psi^{\otimes n}\in\Hilbert^{\otimes n}$ can be transformed into $\varphi_n$ with vanishing probability of failure, where 
    \begin{equation}
    \epsilon_n\coloneqq 1-\lvert\braket{\varphi_n}{\varphi^{\otimes n}}\rvert^2    
    \end{equation}
    vanishes in the asymptotic limit, then $E(\psi)\ge E(\varphi)$.
\end{corollary}
\begin{proof}
    By the weak monotonicity, weak additivity and asymptotic continuity we have 
    \begin{equation}
    \begin{split}
       nE(\psi) &= E(\psi^{\otimes n})\\ &\ge P_n E(\varphi_n) \\&\ge n P_n  E(\varphi)-a\left(\sqrt{\epsilon_n}\right) n\log \dim \Hilbert + b\left(\sqrt{\epsilon_n}\right)
    \end{split}
    \end{equation}
    where $P_n\ge 0$ is such that $\lim_{n\to\infty}P_n=1$ and $a, b$ are vanishing in $\epsilon_n\to 0$. Dividing this by $n$ and taking the limit $n\to\infty$ we get $E(\psi)\ge E(\varphi)$.
\end{proof}

\begin{corollary}\label{cor:smoothliminF}
    Let $E\in\mathcal{F}_{\text{sub},k}$ and assume that $\Phi(E)$ is additive, i.e. $\Phi(E)(\psi_1\otimes\psi_2)= \Phi(E)(\psi_1)+\Phi(E)(\psi_2)$. Then $\frac{\Phi(E)}{\Phi(\GHZ)}\in\mathcal{F}_k$ and the inequality \labelcref{eq:asympContle} holds with equality for this normalized entanglement measure.
\end{corollary}
\begin{proof}

Normalization over the GHZ state \labelcref{it:norm} follows by construction.
Note that any additivity or monotonicity property or \cref{eq:asympContle} is independent of normalization, in fact these properties are kept if one multiplies the entanglement measure $E$ by any $\lambda>0$. Then full additivity \labelcref{it:additive} follows from the assumptions,
monotonicity on average \labelcref{it:monotone} follows from \cref{lem:phiEsubweakmon}, asymptotic continuity was proven in \cref{thm:AEP}. The equality follows from \cite[Proposition 4.1.]{vrana2022asymptotic}
\end{proof}

\section{The asymptotic spectrum of LOCC and its AEP}\label{sec:AEPforknown}

In \cite{jensen2019asymptotic} it was shown that the so called \emph{asymptotic spectrum of LOCC transformations} characterizes the optimal converse transformation rate \eqref{eq:Rconverse}. The asymptotic spectrum of LOCC consists of monotone (keeping LOCC order) functionals $F: S_k\to\reals_{\ge 0}$ that are invariant under local isometries, normalized on unit tensors $\langle r\rangle$ ($r$-level GHZ state), additive under the direct sum, multiplicative under the tensor product, and for some $\alpha\in[0,1]$ they satisfy 
$F(\sqrt{p}\psi)=p^\alpha F(\psi)$.

In this paper, we work with the more convenient form $E^\alpha({\psi})=\frac{\log F^\alpha({\psi})}{1-\alpha}$ of spectrum-elements called the logarithmic spectral points, where $\alpha\in[0,1)$ is given by the scaling property \labelcref{it:EspecScale}. The axioms defining the spectrum can be rephrased straightforwardly to logarithmic spectral functionals as follows.
\begin{enumerate}[({S}1)]\label{en:loccraxioms}
    \item\label{it:EspecScale} $E^\alpha(\sqrt{p}{\psi})=E^\alpha({\psi})+\frac{\alpha}{1-\alpha}\log p$,
    \item\label{it:Eunittensor} $E^\alpha(\unittensor{r})= \frac{1}{1-\alpha}\log r$,
    \item\label{it:EspecMulti} $E^\alpha({\psi}\otimes{\varphi})=E^\alpha({\psi})+E^\alpha({\varphi})$,
    \item\label{it:EspecAddi} $E^\alpha({\psi}\oplus{\varphi})=\max_{\lambda\in [0,1]}\left( \lambda E^\alpha(\ket{\psi})+(1-\lambda) E^\alpha({\varphi})+ \frac{1}{1-\alpha}h(\lambda)\right) $,
    \item\label{it:Emonotonicity} $E^\alpha({\psi})$ is monotone (non-increasing) under LOCC.
\end{enumerate}
\cref{it:EspecAddi} comes from the identity \cite[Eq. (2.13)]{strassen1991degeneration}:
\begin{equation}\label{eq:logsup}
    \log\big(2^{x_1}+2^{x_2}\big)=\max_{\lambda\in [0,1]}\quad \left(\lambda x_1 + (1-\lambda) x_2 + h(\lambda)\right),
\end{equation}
where $x_1, x_2\in\reals$ and $h(\lambda)= - \lambda\log\lambda - (1-\lambda)\log(1-\lambda)$ is the binary entropy. 
Functionals mapping $k$-partite vectors to $\reals$ and satisfying \cref{it:EspecAddi,it:EspecMulti,it:EspecScale,it:Eunittensor,it:Emonotonicity} are called {\it logarithmic spectral points}, we denote their set by $\Delta_{\log, k}$.

It was shown in \cite{jensen2019asymptotic} that the only spectrum-element with $\alpha=1$, is the squared norm $F^1(\psi)=\norm{\psi}^2$, on the other hand we can extend the set of logarithmic spectral points to the limit $\alpha\to 1$ as follows. At this limit properties \labelcref{it:EspecScale,it:Eunittensor} are not sensible, therefore we only consider functionals over the set of unit vectors, normalized to $1$ on the GHZ state (property \labelcref{it:norm}). \labelcref{it:EspecMulti,it:Emonotonicity} remain the same in the $\alpha\to 1$ limit. Property \labelcref{it:EspecAddi} can be rewritten in the $\alpha\to 1$ limit as follows. Let $\psi, \varphi$ be unit vectors, $p\in [0,1]$ and $E^\alpha$ be a functional satisfying all of the \labelcref{it:EspecAddi,it:EspecMulti,it:EspecScale,it:Emonotonicity,it:Eunittensor}. Then by \labelcref{it:EspecScale,it:EspecAddi} we can write
\begin{equation}\label{eq:EadditivityRewritten}
\begin{split}
    E^\alpha(\sqrt{p}{\psi}\oplus\sqrt{1-p}{\varphi})&=\max_{\lambda\in [0,1]}\bigg( \lambda E^\alpha({\psi})+(1-\lambda) E^\alpha({\varphi})+ \frac{1}{1-\alpha}h(\lambda) 
    \\
    &+\frac{\alpha}{1-\alpha} \left( \lambda\log p + (1-\lambda)\log (1-p) \right)\bigg)\\
    &= 
    \max_{\lambda\in [0,1]}\left(
    \lambda E^\alpha({\psi})+(1-\lambda) E^\alpha({\varphi})  + h(\lambda)
    -\frac{\alpha}{1-\alpha}  \binaryrelativeentropy{\lambda}{p}\right),
\end{split}
\end{equation}
where $\binaryrelativeentropy{\lambda}{p}$ denotes the Kullback–Leibler divergence between the binary probability distributions $(\lambda,1-\lambda)$ and $(p,1-p)$.
Taking the limit $\alpha\to 1$, the divergence $\binaryrelativeentropy{\lambda}{p}$ dominates the optimizable expression. From the fact that $\binaryrelativeentropy{\lambda}{p}\ge 0$ iff $\lambda=p$ follows the convergence of $\lambda\to p$, and the form in \labelcref{it:asymptoticCont}. In other words, \labelcref{it:asymptoticCont} can be seen as the $\alpha\to 1$ limit of \labelcref{it:EspecAddi}.

In the bipartite case the logarithmic spectral points are exactly the R\'enyi entanglement entropies $\entropy_\alpha(\Tr_1\vectorstate{\psi})$ for $\alpha\in [0,1]$ while $\mathcal{F}_2$ contains only the entanglement entropy $\entropy(\Tr_1\vectorstate{\psi})=\lim_{\alpha\to 1}\entropy_\alpha(\Tr_1\vectorstate{\psi})$.
Although there exist other entanglement measures which are non-increasing under LOCC channels, the entanglement entropy is the only one that survives the asymptotic limit \cite{vidal2000entanglement} in the bipartite case. This statement can be seen as an asymptotic equipartition property (AEP) of bipartite entanglement measures. 

\subsection{The relation of $\Delta_{\text{log},k}$ to $\mathcal{F}_{\text{sub},k}$}

In this section we observe the relation of the logarithmic spectral points to $\mathcal{F}_{\text{sub},k}$.

\begin{lemma}\label{lem:Enonneg}
    Let $E^\alpha\in\Delta_{\text{log},k}$, where $\alpha$ is determined by its scaling \labelcref{it:EspecScale}. For any unit vector $\psi\in\unitvectors(\Hilbert)$ we have $E^\alpha(\psi)\ge 0$.
\end{lemma}
\begin{proof}
    First note that any pure state $\vectorstate{\psi}$ can be transformed into a product state by an LOCC channel with probability one. This may be done by simply discarding the original state $\psi$ and creating a local state $\ket{0}$ for each party $j$. The resulting state is $\GHZ_1= \ket{0,\dots,0}$. Then \labelcref{it:Eunittensor,it:Emonotonicity} leads to
    \begin{equation}
        E^\alpha(\psi)\ge E^\alpha(\GHZ_1)=0. 
    \end{equation}
\end{proof}

In the following lemma, we formulate the multi-summand version of \labelcref{it:EspecAddi}.
\begin{lemma}\label{lem:Fmultiadditivity}
    Let $E^\alpha\in\Delta_{\text{log},k}$ be a logarithmic spectral point with $\alpha$ scaling, $l\in\naturals$, and for any $i\in [l]$ $\psi_i\in\Hilbert_i$. Then
\begin{equation}
E^\alpha({\psi_1}\oplus\dots\oplus{\psi_l})=\max_{Q\in \distributions([l])}\left( \sum_{i\in [l]} Q_i E^\alpha({\psi_i}) + \frac{1}{1-\alpha}\entropy(Q)\right).
\end{equation}
\end{lemma}
\begin{proof}
The statement holds for $l=2$ by \labelcref{it:EspecAddi}. We use that successively along with the chain rule of Shannon entropy to get the right hand side.

\end{proof}

\begin{lemma}\label{lem:EadditivityMultiSum}
Let $l\in \naturals$ and $P\in\distributions([l])$ with full support on $[l]$. Then any $E^\alpha\in\Delta_{\text{log},k}$ with $\alpha$ scaling satisfies the inequality
\begin{equation}\label{eq:EsumUpperbyHalpha}
    E^\alpha(\sqrt{P(1)}{\psi_1}\oplus\dots\oplus\sqrt{P(l)}{\psi_l})\le
    \max_{i\in[l]} E^\alpha({\psi_i}) + \entropy_\alpha(P)
    .
\end{equation}
In particular we have the upper bound
\begin{equation}\label{eq:EsumsimpleUpper}
    E^\alpha(\sqrt{P(1)}{\psi_1}\oplus\dots\oplus\sqrt{P(l)}{\psi_l})\le
    \max_{i\in[l]} E^\alpha({\psi_i}) + \log l
    .
\end{equation}
\end{lemma}
\begin{proof}
We begin with a calculation analogous to \cref{eq:EadditivityRewritten}. By \cref{lem:Fmultiadditivity} and the scaling property \labelcref{it:EspecScale} we write
\begin{equation}
\begin{split}
    E^\alpha(\sqrt{P(1)}{\psi_1}\oplus\dots\oplus\sqrt{P(l)}{\psi_l})&=\max_{Q\in \distributions([l])}\Bigg( \sum_{i\in [l]} Q(i) E^\alpha({\psi_i}) + \frac{1}{1-\alpha}\entropy(Q)
    \\
    &+\frac{\alpha}{1-\alpha} \sum_{i\in [l]}Q(i)\log P(i) \Bigg) \\
&=\max_{Q\in \distributions([l])} \Bigg(\sum_{i\in [l]} Q(i) E^\alpha({\psi_i}) + \entropy(Q)
    \\
    &-\frac{\alpha}{1-\alpha} \relativeentropy{Q}{P}\Bigg)\\
&\le \max_{i\in [l]} E^\alpha({\psi_i}) + 
\max_{Q\in \distributions([l])} \left(\entropy(Q)
    -\frac{\alpha}{1-\alpha} \relativeentropy{Q}{P}\right).
\end{split}
\end{equation}
The proof of \cref{eq:EsumUpperbyHalpha} is finished via the variational formula \cref{eq:variationalRenyientropy},
from which \cref{eq:EsumsimpleUpper} follows from $\entropy_\alpha\le \log l$.
\end{proof}

In \cite{jensen2019asymptotic}[Proposition 3.3 and 3.4] they show that the spectrum elements can be extended to conditionally pure states and that these are monotone under LOCC. The following lemma is a direct consequence of the this statement, by writing the extension in terms of the logarithmic functionals and using \eqref{eq:logsup} without the supremum as a lower bound. Nevertheless, here we show it directly.

\begin{lemma}\label{lem:EmonotoneOnAvg}
    Let $E^\alpha\in\Delta_{\text{log},k}$. Then if we restrict $E^\alpha$ to unit vectors, the resulting functional is monotone on average (therefore also weakly monotone), i.e., if
    \begin{equation}\label{eq:psiLOCC}
        \psi\LOCC \sum_{x\in\mathcal{X}}
        P(x)\vectorstate{\psi_x}\otimes\vectorstate{x},
    \end{equation}
    then
    \begin{equation}
        E^{\alpha}(\psi)\ge
        \sum_{x\in\mathcal{X}}P(x)E^\alpha(\psi_x).
    \end{equation}
\end{lemma}
\begin{proof}
    Let $\Lambda$ be the the LOCC channel implementing \cref{eq:psiLOCC}. Then we also have
    \begin{equation}
        \psi^{\otimes n}
        \overset{\Lambda^{\otimes n}}{\longrightarrow} 
        \sum_{\overline{x}\in\mathcal{X}^n}
    P(\overline{x})\vectorstate{\psi_{\overline{x}}}\otimes\vectorstate{\overline{x}},
    \end{equation}
where $\overline{x}=(x_1,\dots,x_n)\in\mathcal{X}^n$ are $n$-strings, $P(\overline{x})=P(x_1)\dots P(x_{n})$, and $\psi_{\overline{x}}=\psi_{x_1}\otimes\dots\otimes\psi_{x_n}$.
Choosing a specific string $\overline{x}$, we can transform unitarily (by the permutation of tensor powers) any permutations of the string $\overline{x}$ to $\overline{x}$, therefore having
\begin{equation}
\psi^{\otimes n}\LOCC
{n \choose \#_{x\in\mathcal{X}}}
{
P(x_1)\dots P(x_{n})} \,
\vectorstate{\psi_{\overline{x}}}, 
\end{equation}
where $\#_{x\in\overline{x}}$ denotes the number of occurences of $x\in\mathcal{X}$ in the $n$-string $\overline{x}$, and the multinomial coefficient gives the number of all possible permutations of $\overline{x}$. Then by additivity \labelcref{it:EspecMulti}, the
scaling property \labelcref{it:EspecScale} and the monotonicity \labelcref{it:Emonotonicity} we write
\begin{equation}\label{eq:Ealphamononavgmaineq}
\begin{split}
    E^\alpha(\psi)&=\frac{1}{n}E^\alpha(\psi^{\otimes n}) \\
    &\ge \frac{1}{n}E^\alpha(\psi_{\overline{x}})
    + \frac{1}{n}\frac{\alpha}{1-\alpha}
    \log {n \choose \#_{x\in\overline{x}}}
{
P(x_1)\dots P(x_{n})}.
\end{split}
\end{equation}
The term in the logarithm is known as the probability of the type-class of $\overline{x}$, i.e. the total probability of the strings which are permutations of $\overline{x}$ (see \cite{csiszar2011information}). This is lower bounded by
\begin{equation}
    {n \choose \#_{x\in\overline{x}}}
P(x_1)\dots P(x_{n}) \ge \frac{1}{(n+1)^{\lvert\mathcal{X}\rvert}}2^{-n\relativeentropy{\frac{\#_{x\in\overline{x}}}{n}}{\{P(x)\}_{x\in\mathcal{X}}}},
\end{equation}
 Then
 \begin{equation}
     \limsup_{n\to\infty}\frac{1}{n}\log {n \choose \#_{x\in\overline{x}}}
     P(x_1)\dots P(x_{n})
     \ge  -\limsup_{n\to\infty}\relativeentropy{\frac{\#_{x\in\overline{x}}}{n}}{\{P(x)\}_{x\in\mathcal{X}}},
 \end{equation}
which converges to $0$ if the sequence of $\overline{x}$ is chosen such that the relative frequencies $\frac{\#_{x\in\overline{x}}}{n}$ of each $x\in\mathcal{X}$ converge to $P(x)$. For the first term in \cref{eq:Ealphamononavgmaineq}, using additivity we write
    \begin{equation}
        \frac{1}{n}E^\alpha(\psi_{\overline{x}})=
        \sum_{x\in\mathcal{X}}\frac{\#_{x\in\overline{x}}}{n} E^\alpha(\psi_{x}),
    \end{equation}
    which converges to 
$\sum_{x\in\mathcal{X}}P(x) E^\alpha(\psi_{x})$ in the $n\to\infty$ limit, by the choice of $\overline{x}$.
\end{proof}

\begin{corollary}\label{cor:deltaloginfsub}
    Let us restrict $E\in\Delta_{\text{log},k}$ to the unit vectors (and still denote it by $E$). Then $E\in\mathcal{F}_{\text{sub},k}$.
\end{corollary}
\begin{proof}
    The subadditivity and the normalization on the $\GHZ$ states are satisfied by definition, the rest of the first statement is the direct consequence of the previous lemmas (\cref{lem:Enonneg,lem:EmonotoneOnAvg,lem:EadditivityMultiSum}).
\end{proof}

\subsection{Evaluating the smoothing limit on the known spectral points}\label{sub:eval}

An $\alpha,\{\theta_j\}_j$ parametric family of functionals $E^{\alpha,\theta}\in\Delta_k$ over the $k$-partite vectors (not necessarily unit) was constructed in \cite{vrana2023family}, where $\alpha\in [0,1)$ defines its scaling property and $\theta_j$, $j\in[k]$ are convex weights (probability distribution) over the subsystems.

In this section, we evaluate the smoothing limit on these entanglement measures, and
conclude that the results are elements of $\mathcal{F}_k$. 
We avoid the explicit introduction of these functionals, because we only need the bounds we introduce in the following.

For any fixed weights $\theta_j$, the functionals $E^{\alpha,\theta}$ are non-increasing as functions of $\alpha$.
The minimal elements of this family are
\begin{equation}\label{eq:knownelementslowerbound}
    E^{1,\theta}(\psi)= 
    \sum_{j\in [k]} \theta_j
    \entropy\left(\frac{\Tr_j \vectorstate{\psi}}{\norm{\psi}^2}\right) \le E^{\alpha,\theta}(\psi),
\end{equation}
which are the convex combinations of the von Neumann entropies of the marginals, while the maximal elements are
\begin{equation}\label{eq:knownelementsupperbound}
     E^{\alpha,\theta}(\psi)\le  E^{0,\theta}(\psi)\le
    \sum_{j\in[k]} \theta_j \,\log\rank\Tr_j \vectorstate{\psi}  =
    \sum_{j\in[k]} \theta_j \entropy_0(\Tr_j \vectorstate{\psi})
    ,
\end{equation}
which can be seen from \cite[Definition 4.4]{vrana2023family}, using that an LOCC transformation can not increase the rank of the marginal states.

The common property of these bounds is that they are expressed in terms of convex combinations of bipartite R\'enyi entanglement entropies. In the following lemma we connect this with the classical AEP.
\begin{lemma}\label{lem:quantumAEP}
    Let $\psi\in\unitvectors(\Hilbert_1\otimes\dots\otimes\Hilbert_k)$ be a $k$-partite unit vector, $0\le\theta_j$ and $0\le\alpha_j\le 1$ for $b\subseteq [k]$. Then
    \begin{equation}
        \lim_{\epsilon \to 0}\limsup_{n \to \infty} \frac{1}{n} \inf_{\substack{\varphi \in \unitvectors(\mathcal{H}^{\otimes n})\\ \lvert\langle \varphi \rvert\psi^{\otimes n} \rangle|^2 \geq 1 - \epsilon}}
        \sum_{j\in[k]} \theta_j \entropy_{\alpha_j}\left(\Tr_j \vectorstate{\varphi}\right)=
        \sum_{j\in[k]} \theta_j \entropy\left(\Tr_j \vectorstate{\psi}\right).
    \end{equation}
In particular
    \begin{equation}
        \lim_{\epsilon \to 0}\limsup_{n \to \infty} \frac{1}{n} \inf_{\substack{\varphi \in \unitvectors(\mathcal{H}^{\otimes n})\\ \lvert\langle \varphi \rvert\psi^{\otimes n} \rangle|^2 \geq 1 - \epsilon}}
         \entropy_{\alpha}\left(\Tr_j \vectorstate{\varphi}\right)=
         \entropy\left(\Tr_j \vectorstate{\psi}\right),
    \end{equation}
    for any $\alpha\in [0,1]$ and $j\in[k]$.
\end{lemma}
\begin{proof}

First note that 
\begin{equation}\label{eq:evaluationmainge}
\begin{split}
    \inf_{\substack{\varphi \in \unitvectors(\mathcal{H}^{\otimes n})\\ \lvert\langle \varphi \rvert\psi^{\otimes n} \rangle|^2 \geq 1 - \epsilon}}
        \sum_{j\in[k]} \theta_j \entropy_{\alpha_j}\left(\Tr_j \vectorstate{\varphi}\right)
        &\ge
        \sum_{j\in[k]} \theta_j
        \inf_{\substack{\varphi \in \unitvectors(\mathcal{H}^{\otimes n})\\ \lvert\langle \varphi \rvert\psi^{\otimes n} \rangle|^2 \geq 1 - \epsilon}}
         \entropy_{\alpha_j}\left(\Tr_j \vectorstate{\varphi}\right)\\
         &\ge
        \sum_{j\in[k]} \theta_j
        \inf_{\substack{\varphi \in \unitvectors(\mathcal{H}^{\otimes n})\\ \lvert\langle \varphi \rvert\psi^{\otimes n} \rangle|^2 \geq 1 - \epsilon}}
         \entropy\left(\Tr_j \vectorstate{\varphi}\right)
         .
\end{split}
\end{equation}
The rest of the proof (of this direction) is done by the fact that the von Neumann entropy is asymptotically continuous (Fannes' inequality \cite{audenaert2007sharp}), which gives
\begin{equation}
    \entropy\left(\Tr_j \vectorstate{\varphi}\right) \ge
    \entropy\left(\Tr_j \vectorstate{\psi}^{\otimes n}\right) +\sqrt{\epsilon}\,n\log\dim\Hilbert + h(\sqrt{\epsilon}),
\end{equation}
The difference terms vanish if we divide them by $n$ and take $\lim_{\epsilon \to 0}\limsup_{n \to \infty}$.

To show the converse, first we write
\begin{equation}\label{eq:HinfvarpileHinfP}
\begin{split}
  \inf_{\substack{\varphi \in \unitvectors(\mathcal{H}^{\otimes n})\\ \lvert\langle \varphi \rvert\psi^{\otimes n} \rangle|^2 \geq 1 - \epsilon}}
        \sum_{j\in[k]} \theta_j \entropy_{\alpha_j}\left(\Tr_j \vectorstate{\varphi}\right)
        &\le
        \inf_{\substack{\varphi \in \unitvectors(\mathcal{H}^{\otimes n})\\ \lvert\langle \varphi \rvert\psi^{\otimes n} \rangle|^2 \geq 1 - \epsilon}}
        \sum_{j\in[k]} \theta_j
         \entropy_{0}\left(\Tr_j \vectorstate{\varphi}\right)\\
         &=
        \inf_{\substack{\varphi \in \unitvectors(\mathcal{H}^{\otimes n})\\ \lvert\langle \varphi \rvert\psi^{\otimes n} \rangle|^2 \geq 1 - \epsilon}}
        \sum_{j\in[k]} \theta_j\,
         \rank\left(\Tr_j \vectorstate{\varphi}\right) 
         .
\end{split}
\end{equation}
We aim to show that
\begin{equation}\label{eq:infvarphileinfQ}
    \inf_{\substack{\varphi \in \unitvectors(\mathcal{H}^{\otimes n})\\ \lvert\langle \varphi \rvert\psi^{\otimes n} \rangle|^2 \geq 1 - \epsilon}}  \sum_{j\in[k]} \theta_j  
    \rank(\Tr_j \vectorstate{\varphi})\le
    \sum_{j\in[k]} \theta_j \,
    \inf_{\substack{Q\in 
    B_{{\epsilon/k}}(\Tr_j\vectorstate{\psi})}}
     \lvert\supp Q\rvert.
\end{equation}
To do this, we show that for any $k$-tuple of probability distributions $(Q_1,\dots,Q_{k})$, such that $Q_j\in B_{{\epsilon/k}}(\Tr_j\vectorstate{\psi})$, there is a suitable vector $\varphi$, such that $\rank(\Tr_j \vectorstate{\varphi})\le \lvert \supp Q_j\rvert$ for each $j\in[k]$. 
In fact it is enough to consider the probability distributions $Q_j\in B_{{\epsilon/k}}(\Tr_j\vectorstate{\psi})$ with the minimal possible support. Those can be achieved by starting from the probability distribution consisting of the eigenvalues of $\Tr_j\vectorstate{\psi}$ and reassigning the lowest probability values up to ${ {\epsilon/k}}$ total probability. 

Let 
\begin{equation}
    \psi^{\otimes n}= \sum_{i=1}^{r_j} \sqrt{\lambda^{(j)}_i} e^{(j)}_i \otimes f^{(j)}_i
\end{equation}
be the Schmidt decomposition of $\psi^{\otimes n}$ over the $j$-th bipartition, where $\{e^{(j)}_i \}_i$ and $\{f^{(j)}_i \}_i$ are orthonormal systems on $\Hilbert_j^{\otimes n}$ and $\bigotimes_{j'\neq j} \Hilbert_{j'}^{\otimes n}$ respectively, and $\lambda^{(j)}_i>0$ are the decreasingly ordered eigenvalues of $\Tr_j\vectorstate{\psi}^{\otimes n}$ (possibly containing multiplicities). Let
\begin{equation}
    P\coloneqq \left(\sum_{i=1}^{r'_1} \vectorstate{e^{(1)}_i} \right)\otimes\dots\otimes\left(
    \sum_{i=1}^{r'_k} \vectorstate{e^{(k)}_i}
    \right),
\end{equation}
where $r'_j\in\naturals$ are the minimal integers such that $\sum_{i=r'_j+1}^{r_j} \lambda^{(j)}_i \le \frac{\epsilon}{k}$. By the union bound (applied to intersections) we have
\begin{equation}
    \norm{P\psi^{\otimes n}}^2 \ge 1 - \sum_{j=1}^k \norm{\left(I-\sum_{i=1}^{r'_j} \vectorstate{e^{(j)}_i}\right)\psi^{\otimes n}}^2 \ge 1- \sum_{j=1}^k \frac{\epsilon}{k}\eqqcolon 1-\epsilon.
\end{equation}
Then also ${\textstyle\braket{\psi^{\otimes n}}{\frac{P\psi^{\otimes n}}{\norm{P\psi^{\otimes n}}}}\ge 1- \epsilon}$, therefore $\varphi\coloneqq P\psi^{\otimes n}/\norm{P\psi^{\otimes n}}$ admits the condition of the infimum.

The Schmidt decomposition of $P\psi^{\otimes n}$ over the $j$-th bipartition is 
\begin{equation}
\begin{split}
    P\psi^{\otimes n}&=\sum_{i=1}^{r_j} \sqrt{\lambda^{(j)}_i}  \left(\sum_{i'=1}^{r'_j} \vectorstate{e^{(j)}_{i'}} \right)e^{(j)}_i \otimes  \left(\bigotimes_{l\neq j}\sum_{i'=1}^{r'_l} \vectorstate{e^{(l)}_{i'}} \right) f^{(j)}_i\\
    &=
    \sum_{i=1}^{r'_j} \sqrt{\lambda^{(j)}_i}  e^{(j)}_i \otimes  \left(\bigotimes_{l\neq j}\sum_{i'=1}^{r'_l} \vectorstate{e^{(l)}_{i'}} \right) f^{(j)}_i,
\end{split}
\end{equation}
which has Schmidt rank at most $r'_j$, i.e. $\rank\Tr_j\vectorstate{\varphi}\le r'_j$. 
On the other hand, $r'_j$ is by definition the lowest achievable support size of probability distributions attainable by reassigning up to $\epsilon/k$ total probability in the distribution consisting of the eigenvalues of $\Tr_j\vectorstate{\psi}$. In other words, for any $j\in [k]$ and $Q\in B_{{\epsilon/k}}(\Tr_j\vectorstate{\psi})$, we have $\rank(\Tr_j\vectorstate{\varphi})\le\supp Q$ for the constructed $\varphi$, which implies \cref{eq:infvarphileinfQ}.

The proof is finished by dividing \cref{eq:HinfvarpileHinfP} by $n$ and taking the limsup $n\to\infty$, then $\epsilon\to 0$.
By the (classical) AEP (\cref{eq:AEPentropic}) we have
\begin{equation}
    \lim_{\epsilon\to 0}\limsup_{n\to\infty}\frac{1}{n}
    \sum_{j\in[k]} \theta_j
    \inf_{Q\in B_{{\epsilon/k}}\left(\Tr_j \vectorstate{\psi}^{\otimes n}\right)} 
    \entropy_0(P)=
    \sum_{j\in[k]} \theta_j
    \entropy\left(\Tr_j \vectorstate{\psi}^{\otimes n}\right).
\end{equation}

\end{proof}

\begin{corollary}\label{prop:knownmeasuresPhi0}
The smoothing limits of $E^{\alpha,\theta}$ for $\alpha\in[0,1)$ and $\theta\in\distributions([k])$, introduced in \cite{vrana2023family} is 
\begin{equation}\label{eq:smoothofknown}
\Phi(E^{\alpha,\theta})(\psi)=\sum_{j\in[k]} \theta_j \entropy\left(\frac{\Tr_j\vectorstate{\psi}}{\norm{\psi}^2}\right).
\end{equation}
These entanglement measures are elements of $\mathcal{F}_k$
\end{corollary}
\begin{proof}
    The first statement is the direct consequence of \cref{lem:quantumAEP} combined with the bounds  \cref{eq:knownelementslowerbound,eq:knownelementsupperbound}.
    The second statement for the functionals $E^{\alpha,\theta}$ is the consequence of \cref{cor:smoothliminF}, since the von Neumann entropies are additive, and the marginal of the $\GHZ$ state by any non-trivial bipartition is the maximally mixed state on a 2-dimensional subspace, with von Neumann entropy $1$. Alternatively, the axioms \labelcref{it:additive,it:asymptoticCont,it:monotone,it:norm} can be easily verified directly.

\begin{remark}
Note that any convex combination of the von Neumann enropies is additive, asymptotically continuous and monotone on average. Then we can extend the family of entanglement measures in \eqref{eq:smoothofknown} to a larger subset of $\mathcal{F}_k$ by
    \begin{equation}
    \entropy^\theta(\psi)\coloneqq
        \sum_{b\subseteq[k]} \theta_b \entropy\left(\frac{\Tr_b\vectorstate{\psi}}{\norm{\psi}^2}\right),
    \end{equation}
    where $b$ runs over all of the bipartitions of the $k$-partite system.

\end{remark}

\end{proof}

\section{Final remarks}

We have shown that subadditive entanglement measures ($\mathcal{F}_{k,\text{sub}}$) admitting certain weak condition (\cref{def:deltasub}) are subject to an asymptotic equipartition property, in the form of the smoothing limit having an idempotent action on them. The image of the smoothing map consists of the subadditive, weakly monotone and asymptotically continuous entanglement measures. It is not clear whether the measures in the image are automatically monotone on average, therefore being the part of $\mathcal{F}_{k,\text{sub}}$ or not. On the other hand, the additive part of the image is exactly $\mathcal{F}_k$ (up to normalization), the entanglement measures characterizing $\text{LOCC}_q$.

We obtained elements of $\mathcal{F}_k$ when we evaluated the regularization of the smoothed version of the family of multipartite measures $E^{\alpha,\theta}$ (elements of the asymptotic spectrum of LOCC) introduced in \cite{vrana2023family}. Still we do not know in general, if all the elements of the image of $\Phi$ are additive or not, in other words if $\mathcal{F}_k$ is a proper subset of the image of $\Phi$ or if these sets are equal. The diagram below summarizes the
relations between the different entanglement measures considered in this work, also indicating the open questions.

\begin{center}
\begin{tikzpicture}[scale=1, every node/.style={scale=1}]
  \node (top) at (0, 2.2) {$\mathcal{F}_{k,\text{sub}}$};
  \node (phi1) at (-0.6, 1.2) {$\Phi[\mathcal{F}_{k,\text{sub}}]$};
  \node (phi2) at (-1.3, 0) {$\Phi[\Delta_{k,\text{sub}}]$};
  \node (sum) at (-1.3, -1.2) {$\{\entropy^\theta\}_{\supp\theta\subseteq [k]}$};
    \node (sum) at (1.55, -1.2) {$\{\entropy^\theta\}_{\supp\theta\subseteq 2^{[k]}}$};
  \node (f2) at (0.75, 0) {$\mathcal{F}_k$};
  \node (delta) at (2.6, 1.3) {$\Delta_{k,\text{sub}}$ (on unit vectors)};
  \node (arrowtext) at (1.35, 0.7) {\small $\alpha \to 1$};

  \node at (-0.85, 1.9) {\rotatebox{45}{$\overset{\Large ?}{\subsetneq}$}};
  \node at (-1.5, 0.7) {\rotatebox{45}{$\subseteq$}};
  \node at (-1.95, -0.6) {\rotatebox{90}{$\in$}};
  \node at (0.7, -0.6) {\rotatebox{90}{$\in$}};
  \node at (0.8, 0.7) {\rotatebox{90}{$\longleftarrow$}};
  \node at (0.6, 1.75) {\rotatebox{-45}{$\supsetneq$}};
  \node at (0.05, 0.6) {\rotatebox{-45}{$\supseteq\overset{?}{=}$}};
  \node at (0.1, -1.2) {\rotatebox{0}{$\subseteq$}};
\end{tikzpicture}
\end{center}

\section*{Acknowledgement}

I thank Péter Vrana for the discussions and his insightful advice. 

This work was partially funded by the National Research, Development and Innovation Office of Hungary (NKFIH) via the research grants FK 146643,
K 146380, and EXCELLENCE 151342, and by the Ministry of Culture and Innovation and the National Research, Development and Innovation Office within the Quantum
Information National Laboratory of Hungary (Grant No. 2022-2.1.1-NL-2022-00004). 

\appendix

\section{Technical lemmas}

\begin{lemma}\label{lem:Tboundreformulatedelementwise}
    For a finite set $\mathcal{X}$, let $P(x), Q(x)\in\distributions(\mathcal{X})$ and $\psi_x, \varphi_x\in\unitvectors(\Hilbert)$ be unit vectors for each $x$. Also let
    \begin{equation*}
        \rho=\sum_{x\in\mathcal{X}} P(x) \vectorstate{\psi_x}\otimes\vectorstate{x}
    \end{equation*}
    and
    \begin{equation*}
        \sigma=\sum_{x\in\mathcal{X}} Q(x) \vectorstate{\varphi_x}\otimes\vectorstate{x}
    \end{equation*}
    be conditionally pure states in $\Hilbert\otimes\Diag(\complexes^{ \mathcal{X} })$, such that $\text{T}(\rho,\sigma)\le \epsilon$.
Then
 \begin{equation}\label{eq:PQvariationdistanceupperb}
    \TV(P,Q)=\frac{1}{2}\sum_{x\in\mathcal{X}} \lvert P(x)-Q(x)\rvert\le \epsilon,
 \end{equation}
 and
\begin{equation}
    \sum_{x\in\mathcal{X}}P(x)\text{T}(\vectorstate{\psi_x},\vectorstate{\varphi_x}) \le 2\epsilon.
\end{equation}

\end{lemma}
 \begin{proof}
     We start with
     \begin{equation}\label{eq:Tforclassicalquantum}
     \begin{split}
         \text{T}(\rho,\sigma)&=
         \frac{1}{2}\norm[1]{\sum_{x\in\mathcal{X}} P(x) \vectorstate{\psi_x}\otimes\vectorstate{x}-\sum_{x\in\mathcal{X}} Q(x) \vectorstate{\varphi_x}\otimes\vectorstate{x}}\\
         &=
         \frac{1}{2}\norm[1]{\sum_{x\in\mathcal{X}}\left( P(x) \vectorstate{\psi_x}- Q(x) \vectorstate{\varphi_x}\right)\otimes\vectorstate{x}}\\
         &=
         \frac{1}{2}\sum_{x\in\mathcal{X}}\norm[1]{ P(x) \vectorstate{\psi_x}- Q(x) \vectorstate{\varphi_x}}.
     \end{split}
     \end{equation}
Next we use the triangle inequality to write
\begin{equation}
\begin{split}
    \text{T}(\rho,\sigma)&=
         \frac{1}{2}\sum_{x\in\mathcal{X}}\norm[1]{ P(x) \vectorstate{\psi_x}- P(x) \vectorstate{\varphi_x} 
         + (P(x)-Q(x))\vectorstate{\varphi_x}}\\
         &\ge
         \frac{1}{2}\sum_{x\in\mathcal{X}}\bigg(P(x)\norm[1]{\vectorstate{\psi_x}- \vectorstate{\varphi_x}}
         -\lvert P(x)-Q(x) \rvert
         \bigg)\\
         &=
         \frac{1}{2}\sum_{x\in\mathcal{X}}\left(P(x)\norm[1]{\vectorstate{\psi_x}- \vectorstate{\varphi_x}}\right)
-\frac{1}{2}\sum_{x\in\mathcal{X}}\lvert P(x)-Q(x)\rvert.
\end{split}
\end{equation}
To finish the proof of both statements, we need to bound the second term by $\epsilon$.
This is done by the fact that $2\text{T}(\rho,\sigma)\ge \Tr (\rho-\sigma)=\sum_{x\in\mathcal{X}} P(x)-Q(x)$, then by assumption $\sum_{x\in\mathcal{X}} P(x)-Q(x) \le 2\epsilon$ and also $\sum_{x\in\mathcal{X}} Q(x)-P(x) \le 2\epsilon$.

 \end{proof}

\begin{lemma}\label{lem:LOCCoutcomeproblower}
    Let $\psi, \varphi\in\unitvectors(\Hilbert)$ such that $\lvert\braket{\varphi}{\psi}\rvert^2\ge 1-\epsilon$  and $\Lambda$ is an LOCC channel transforming the pure state $\vectorstate{\psi}$ as
    \begin{equation}
        \vectorstate{\psi}
\overset{\Lambda}{\longrightarrow} 
\sum_{x \in \mathcal{X}} P(x) 
\vectorstate{\psi_x}  \otimes \vectorstate{x},
    \end{equation}
    and similarly 
    \begin{equation}
        \vectorstate{\varphi}
\overset{\Lambda}{\longrightarrow} 
\sum_{x \in \mathcal{X}} Q(x) 
\vectorstate{\varphi_x}  \otimes \vectorstate{x},
    \end{equation}
    where the support of the probability distributions $P$ and $Q$ may differ. Then
\begin{equation}
    \sum_{\substack{x\in\mathcal{X}\\\lvert\braket{\varphi_x}{\psi_x}\rvert^2\ge 1-\epsilon'}} 
     Q(x)\ge   
     1
     -2\sqrt{\epsilon}\left(1+\frac{1}{\sqrt{\epsilon'}}\right).
\end{equation}
\end{lemma}
    
\begin{proof}
    We rewrite the conditon for the overlap using \cref{eq:tracedistofpure}, then relax it using the data processing inequality:
\begin{equation}
\begin{split}
    \sqrt{\epsilon} &\ge \text{T}\left(\vectorstate{\varphi},\vectorstate{\psi}\right)\\&\ge 
    \text{T}\left(\Lambda(\vectorstate{\varphi}),\Lambda(\vectorstate{\psi})\right)
    \\&=
    \text{T}\left(\sum_{{x} \in \mathcal{X}} Q({x}) 
\vectorstate{\varphi_{{x}}}  \otimes \vectorstate{{x}},
\sum_{{x} \in \mathcal{X}} P({x}) 
\vectorstate{\psi_{{x}}}\otimes \vectorstate{{x}}\right).
\end{split}
\end{equation}
For the sake of brevity, let $\T_{{x}}\coloneqq\T(\vectorstate{\varphi_{{x}}},\vectorstate{\psi_{{x}}})$.
By \cref{lem:Tboundreformulatedelementwise} we have
\begin{equation}    
\begin{split}
2\sqrt{\epsilon}&\ge
\sum_{{x} \in \mathcal{X}}P({x}) \T_{{x}}
\\
&\ge
\sum_{\substack{{x} \in \mathcal{X}\\\T_{{x}}>\sqrt{\epsilon'}}}P({x})\T_{{x}} \\
&\ge
\sqrt{\epsilon'}\sum_{\substack{{x} \in \mathcal{X}\\\T_{{x}}>\sqrt{\epsilon'}}}P({x})
.
\end{split}
\end{equation}

Again by \cref{lem:Tboundreformulatedelementwise} and the previous bound  we have
\begin{equation}
    \sum_{\substack{{x}\in\mathcal{X}\\ T_x\le \sqrt{\epsilon'}}} 
     Q({x})\ge \sum_{\substack{{x}\in\mathcal{X}\\ T_x\le \sqrt{\epsilon'}}} 
     P({x}) -2\sqrt{\epsilon}
     \ge  
     1 -2\frac{\sqrt{\epsilon}}{\sqrt{\epsilon'}} -2\sqrt{\epsilon}.
\end{equation}
\end{proof}

Next, we show that if the states $\varphi$ and $\psi^{\otimes n}$ are close in trace distance, then their maximal eigenvalues are also close.
\begin{lemma}\label{lem:rhosigmamaxeigenvalues}
    Let $\rho, \sigma \in \boundeds(\mathcal{H})$ be states and let $T(\rho,\sigma)\coloneqq\frac{1}{2}\Tr \lvert \rho-\sigma
    \rvert$ denote their trace distance. Assume that $T(\rho,\sigma)\le \epsilon$, then 
    \begin{equation}
        \lvert \lambda_{\rho,\text{max}}-\lambda_{\sigma,\text{max}} \rvert  \le 2\epsilon 
    \end{equation}
    where $\lambda_{\rho,\text{max}}$ denotes the largest eigenvalue of $\rho$.
\end{lemma}
\begin{proof}
   Without loss of generality, assume $\lambda_{\rho,\text{max}}\ge \lambda_{\sigma,\text{max}}$.
   Let $v$ be a unit eigenvector of $\rho$ with eigenvalue $\lambda_{\rho,\text{max}}$.
    We have the following chain of lower bounds for the trace distance:
    \begin{equation}
    \begin{split}
        T(\rho,\sigma)&\ge
    \frac{1}{2} \bra{v} 
         \left(\rho-\sigma \right)
    \ket{v}\\
    &=\frac{1}{2}\left(
    \lambda_{\rho,\text{max}}-\bra{\psi}
    \sigma
    \ket{\psi}
    \right)\\
    &\ge
    \frac{1}{2}\left(
    \lambda_{\rho,\text{max}}-
    \lambda_{\sigma,\text{max}}
    \right)
    \end{split}
    \end{equation}.
This inequality also holds with $\rho$ and $\sigma$ being swapped.
    \end{proof}

\bibliography{references}

\end{document}